\documentclass[acmsmall]{acmart}

\usepackage[utf8]{inputenc} 
\usepackage[T1]{fontenc}    
\usepackage{hyperref}       
\usepackage{url}            
\usepackage{booktabs}       
\usepackage{amsfonts}       
\usepackage{nicefrac}       
\usepackage{microtype}      
\usepackage{lipsum}
\usepackage{graphicx}
\usepackage{subfigure}
\graphicspath{ {./images/} }
\usepackage{amsmath,amsthm}
\usepackage{algorithm}
\usepackage{algorithmicx}
\usepackage{algpseudocode}
\usepackage{braket}
\usepackage{booktabs}
\usepackage{threeparttable}
\usepackage{xcolor}
\usepackage{supertabular}
\usepackage{longtable}
\usepackage{multirow}
\usepackage{tikz}
\usetikzlibrary{quantikz}

\newtheorem{theorem}{Theorem}
\newtheorem{lemma}{Lemma}
\newtheorem{corollary}{Corollary}
\newtheorem{definition}{Definition}
\newtheorem{example}{Example}
\newtheorem{remark}{Remark}

\newtheorem*{lemma*}{Lemma}
\newtheorem*{theorem*}{Theorem}

\newcommand {\hadamard} {{\textsf{H}}}
\newcommand {\Z} {{\textsf{Z}}}
\newcommand {\T}{{\textsf{T}}}

\newcommand {\hadadmard}{{\textsf{H}}}

\newcommand {\cx} {{\textsf{CX}}}

\newcommand{\bphi}{{\Phi}}
\newcommand {\TDD} {{\mathit{TDD}}}
\newcommand {\cnot} {{\textsf{CX}}}

\newcommand {\F} {{\mathcal{F}}}
\newcommand {\M} {{\mathcal{M}}}
\newcommand {\g} {{\mathcal{G}}}
\newcommand {\cX}[1] {{\mathcal{#1}}}

\newcommand {\cont} {{\textsf{cont}}}

\newcommand{\blue}[1]{{\color{black}{#1}}}

\newcommand{\xblue}[1]{{\color{black}{#1}}}
\newcommand{\rmagenta}[1]{{\color{black}{#1}}}
\newcommand{\rblue}[1]{{\color{black}{#1}}}
\newcommand{\rred}[1]{{\color{black}{#1}}}
\newcommand{\rzxz}[1]{\color{black}}
\usepackage[normalem]{ulem}

\begin{document}

\title{A Tensor Network based Decision Diagram for Representation of Quantum Circuits}


\author{Xin Hong}
\affiliation{%
  \institution{University of Technology Sydney}
  \country{AUS}
}

\author{Xiangzhen Zhou}
\affiliation{%
  \institution{University of Technology Sydney}
  \country{AUS} 
  \institution{Southeast University}
  \country{CN}
}

\author{Sanjiang Li}
\affiliation{%
  \institution{University of Technology Sydney}
  \country{AUS}
}
\email{sanjiang.li@uts.edu.au}

\author{Yuan Feng}
\affiliation{%
  \institution{University of Technology Sydney}
  \country{AUS}
}
\email{yuan.feng@uts.edu.au}

\author{Mingsheng Ying}
\affiliation{%
    \institution{Institute of Software, Chinese Academy of Sciences}
  \country{CN}
    \institution{Tsinghua University}
  \country{CN}
}
\email{yingms@ios.ac.cn} \email {yingmsh@tsinghua.edu.cn}

\begin{abstract}
Tensor networks have been successfully applied in simulation of quantum physical systems for decades. Recently, they have also been employed in classical simulation of quantum computing, in particular, random quantum circuits. This paper proposes a decision diagram style data structure, called TDD (Tensor Decision Diagram), for more
principled and convenient applications of tensor networks. This new data structure provides a compact and canonical representation for quantum circuits. By exploiting circuit partition, the TDD of a quantum circuit can be computed efficiently. Furthermore, we show that the operations of tensor networks essential in their applications (e.g., addition and contraction) can also be implemented efficiently in TDDs. A proof-of-concept implementation of TDDs is presented and its efficiency is evaluated on a set of benchmark quantum circuits. It is expected that TDDs will play an important role in various design automation tasks related to quantum circuits, including but not limited to equivalence checking, error detection, synthesis, simulation, and verification.
\end{abstract}

\setcopyright{acmcopyright}
\acmJournal{TODAES}
\acmYear{2022} \acmVolume{1} \acmNumber{1} \acmArticle{1} \acmMonth{1} \acmPrice{15.00}\acmDOI{10.1145/3514355}



\maketitle
\section{Introduction}
Google's recent demonstration of quantum supremacy on its 53-qubit quantum processor Sycamore \cite{GoogleQsupr} has confirmed that quantum computers can indeed complete tasks much more efficiently than the most advanced traditional computers. Quantum devices of
similar sizes have also been developed at IBM, Intel, IonQ, and Honeywell. It is widely believed that quantum processors with several hundreds of qubits will very likely appear in the next 5-10 years. 
The rapid growth of the size of quantum computing hardware motivates people to develop effective techniques for synthesis, optimisation, testing and verification of quantum circuits.  

Mathematically, quantum circuits  can be represented as unitary matrices, which transform initial quantum states (represented as vectors) to desired output states.  The size of this matrix representation grows exponentially with the size of the quantum system, which makes it a great challenge to even simulate a quantum random circuit with a modest size and a shallow depth. Existing matrix-based packages like Qiskit (https://qiskit.org/) and the Google TensorNetwork \cite{roberts2019tensornetwork}, though very efficient, store such a matrix as a complete array, whose size quickly exceeds the memory limit. For example, it requires 64GB memory to store the functionality of a 16-qubit quantum circuit if each matrix entry is represented in data type \emph{complex128}. 

In order to alleviate the challenge and to provide a compact, canonical, and efficient representation for quantum functionalities, several decision diagram style data structures have been proposed, including Quantum Information Decision Diagrams (QuIDDs) \cite{viamontes2003improving}, Quantum Multiple-Valued Decision Diagrams (QMDDs) \cite{niemann2015qmdds}, \xblue{and Decision Diagrams for Matrix Functions (DDMFs) \cite{yamashita2008ddmf}}. QuIDD is a variant of Algebraic Decision Diagrams (ADDs) \cite{bahar1997algebric} by restricting values to complex numbers, which are indexed by integers, and interleaving row and column variables in the variable ordering. In contrast, QMDD partitions a transformation matrix into four sub-matrices of equal size, which in turn are partitioned similarly, and uses shared nodes to represent sub-matrices differing in only a constant coefficient. 
Evaluations in  \cite{niemann2015qmdds} showed that QMDDs offer a compact representation for large unitary (transformation) matrices. Consequently, they provide a compact and canonical representation for the functionality of quantum circuits. Indeed, QMDDs have been successfully used in  simulation \cite{zulehner2018advanced} and equivalence checking \cite{burgholzer2020improved,burgholzer2020advanced} of quantum circuits as well as verifying the correctness of quantum circuits compilation \cite{wille2020efficient}. \xblue{DDMFs provide an efficient verification scheme for the class of so-called semi-classical quantum circuits \cite{yamashita2008ddmf}.
Very recently, two new decision diagram methods \cite{tsai2020bit,vinkhuijzen2021limdd} have been proposed, both demonstrating superiority than QMDD in some simulation tasks. Note that these methods cannot be directly used for representing the unitary matrix of a quantum circuit. 
}

Tensor networks provide a flexible way to represent quantum circuits and have been successfully employed in the classical simulation of quantum computing in the last few years. By observing that quantum circuits are a special class of tensor networks, Pednault et al. \cite{pednault2017breaking} exploited the flexibility of tensor computations with circuit partition and tensor slicing methods, and broke the 49-qubit barrier of that time in the simulation of quantum circuits. Later on, the size and depth of quantum circuits which can be simulated employing tensor network and the simulation time have been significantly improved (see, e.g., \cite{boixo2017simulation, li2019quantum, chen2018classical, chen201864, huang2020classical}). Tensor networks can also be applied in computing the functionality of a quantum circuit. Indeed, it can be computed in essentially any order, which in turn greatly affects the calculation efficiency. For a quantum circuit with low tree-width, by exploiting an optimal contraction order, the tensor representation of the quantum circuit can be computed in time polynomial in the size of the circuit \cite{markov2008simulating}. While it is in general NP-hard to find an optimal contraction order, one may exploit heuristics like circuit partition \cite{pednault2017breaking}, tree decomposition \cite{markov2008simulating}, and hyper-optimisation approaches \cite{gray2020hyper}, which have been demonstrated as very useful for simulating quantum circuits.

Inspired by the success of tensor networks in the classical simulation of quantum circuits, this paper aims to introduce a novel decision diagram, called Tensor Decision Diagram (TDD for short), for tensor networks. As a new data structure, TDD can further 
explore the flexibility of tensor networks in a more principled way, while overcoming the serious memory bottleneck of matrix-based representations.




While it is observed that the Boole-Shannon expansion commonly used in the design of decision diagrams is ``not a basic decomposition for quantum mechanical phenomena \cite{niemann2015qmdds}'',  tensors, as multidimensional linear maps with complex values, do enjoy Boole-Shannon style expansions. This observation lays the foundation of our design of TDD. 

 
TDDs have several important features that warrant their applicability. Analogous to reduced ordered binary decision diagrams (ROBDD) for Boolean functions \cite{bryant1986graph}, redundant nodes or nodes representing the same tensor in a TDD can be removed or merged so that shared nodes are used as much as possible. The canonicity result  (Theorem~\ref{thm:canonicity}) guarantees that, up to variable ordering, each quantum circuit has a unique reduced TDD representation. 
An efficient algorithm (Alg.~\ref{TDD_generate}) is also designed to generate the reduced TDD representation of a quantum functionality (e.g., a quantum gate or a part of a quantum circuit). Moreover, we show that basic TDD operations such as addition and contraction can be implemented efficiently.  As QMDD, TDD provides a universal, compact and canonical representation for quantum circuits, which is vital in various design automation tasks. \xblue{Indeed, as we shall see in Section~\ref{sec:qmdd_vs._tdd}, TDDs can be regarded as QMDDs under flexible variable order relaxation, which, when combined with tensor network techniques, may potentially provide greater flexibility in computing the final decision diagram representations of  quantum circuits. }


In the remainder of this paper, after a brief review of quantum circuits and QMDD in Sec.~\ref{sec:background} and of tensor networks in Sec.~\ref{sec:tensor}, we introduce our new data structure TDD \xblue{and compare it with QMDD} in Sec.~\ref{sec:tdd}. The construction and implementation of basic tensor operations are presented in Sec.~\ref{sec:construction}. After that, we show how to compute the TDD representation of a quantum circuit in a circuit partition way in Sec. \ref{sec:circuitpartition}. Experimental results are \rmagenta{reported} and analysed in Sec.~\ref{sec:experiments}. The last section concludes the paper and briefly discusses several topics for future research. Most technical proofs 
are presented in the appendices. 

\section{Background} \label{sec:background}
For convenience of the reader, we review some basic concepts about quantum circuits and the Quantum Multi-value Decision Diagram (QMDD) in this section.

\subsection{Quantum Circuits}
The most basic concept in quantum computing is the qubit, which is the counterpart of bit in classical computing. The state of a qubit is often represented in  Dirac notation
\begin{equation}
\label{eq:qubit}
  \ket{ \varphi} :=  \alpha_0 \ket{0} + \alpha_1 \ket{1},
\end{equation}
where $\alpha_0$ and $\alpha_1$ are complex numbers, called the amplitudes of $\ket{\varphi}$, and satisfy ${\left| \alpha_0  \right|^2} + {\left| \alpha_1  \right|^2} = 1$. We also use the vector $[\alpha_0,\alpha_1]^\intercal$ to represent a single-qubit state. In general, an $n$-qubit quantum state can be represented as a $2^n$-dimensional complex vector  $[\alpha_0,\alpha_1,\dots,\alpha_{2^n-1}]^\intercal$.


\begin{figure}
\centerline{$\mathit{H\ gate}: \ \ 
\begin{tikzcd}[column sep=0.4cm,row sep=0.4cm]
&\qw  & \gate{H}  &\qw &\qw \\
\end{tikzcd}\ \ \ \ \ \ \ 
\frac{1}{\sqrt{2}}\left[\begin{array}{cccc} 
		1 & 1\\ 
		1 & -1\\
\end{array}\right]
$}
\centerline{
$\mathit{T\ gate}: \ \ 
\begin{tikzcd}[column sep=0.4cm,row sep=0.4cm]
&\qw  & \gate{T}  &\qw &\qw \\
\end{tikzcd}\ \ \ \ \ \ \ \ \ \ \ 
\left[\begin{array}{cccc} 
		1 & 0\\ 
		0 & e^{\frac{i\pi}{4}}\\
\end{array}\right]
$}
\centerline{
$\mathit{CX\ gate}: 
\begin{tikzcd}[column sep=0.4cm,row sep=0.4cm]
&\qw  & \ctrl{1}  &\qw &\qw \\
&\qw  & \gate{X}  &\qw &\qw \\
\end{tikzcd}\ \ \ \ \ 
\left[\begin{array}{cccc} 
		1 & 0 &0 &0\\ 
		0 & 1 &0 &0\\
		0 & 0 &0 &1\\
		0 & 0 &1 &0\\
\end{array}\right]
$}
\caption{The matrix representations of the \hadadmard, \T, and \cnot gate.}
\label{fig:gates}
\end{figure}

The evolution of a quantum system is described by a unitary transformation. In quantum computing, it is usually called a quantum gate. A quantum gate has a unique unitary matrix representation in a predefined orthonormal basis. Fig.~\ref{fig:gates} shows several such examples. 
The state after applying a specific transformation can be obtained by multiplying the corresponding unitary matrix and the vector that represents the input quantum state. For example, the output state resulted from applying a Hadamard gate to an input state $\left[{\alpha _0}, {\alpha _1}\right]^\intercal$
is calculated as follows
\[\frac{1}{{\sqrt 2 }}\left[ {\begin{array}{*{20}{r}}
1&1\\
1&{ - 1}
\end{array}} \right]\left[ {\begin{array}{*{20}{c}}
{{\alpha _0}}\\
{{\alpha _1}}
\end{array}} \right] = \frac{1}{{\sqrt 2 }}\left[ {\begin{array}{*{20}{c}}
{{\alpha _0} + {\alpha _1}}\\
{{\alpha _0} - {\alpha _1}}
\end{array}} \right].\]
More generally, an $n$-qubit quantum gate is represented as a $2^n\times 2^n$-dimensional unitary transformation matrix.

A quantum circuit consists of a set of qubits and a sequence of elementary quantum gates. Given an input state to the qubits involved, the quantum gates in a quantum circuit will be applied to the input state in a sequential manner. The functionality of an $n$-qubit quantum circuit can also be described by a $2^n\times 2^n$-dimensional unitary transformation matrix.

\subsection{Quantum Multi-value Decision Diagram}\label{sec:qmdd}

Quantum Multi-value Decision Diagram (QMDD) \cite{miller2006qmdd} is a decision diagram based data structure which provides a compact and canonical representation for quantum states and transformation matrices.  

The main idea of QMDD is to recursively partition a $2^n\times 2^n$ transformation matrix $A$ into sub-matrices till matrix elements are reached. The QMDD of $A$ is constructed as follows: First, we introduce a root node, representing the original matrix. The root node has four successors, denoting the  sub-matrices obtained by  partitioning $A$ into four with the same size. Each child node is then further expanded in the same manner. Suppose, in some step, a node corresponding to a matrix element is obtained. Then this node is regarded as 
a terminal node labelled 1, while its corresponding matrix element will be assigned as the weight of its incoming edge. The obtained decision diagram may have redundant nodes and weight-0 edges. After proper normalisation and reduction, we have a reduced decision diagram representation of $A$, which is unique up to the order of variables.

\begin{figure}
\centering
\includegraphics[width=0.2\textwidth]{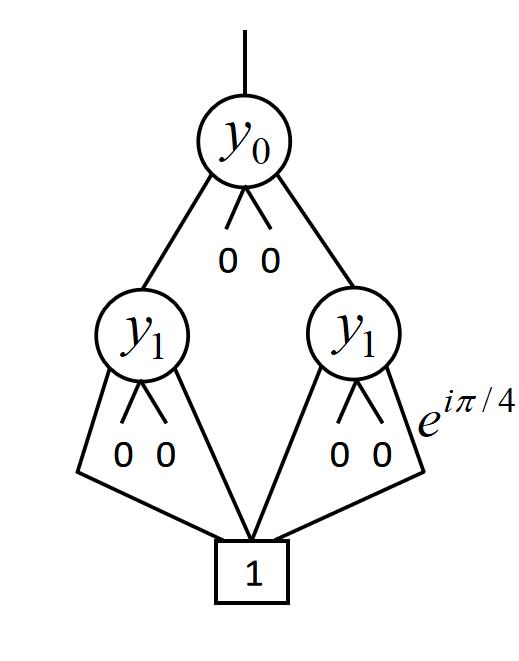}
\caption{The QMDD representation of the controlled-T gate, where the weight of an edge is omitted if it is 1.}
\label{fig:qmdd}
\end{figure}

\begin{example}
Shown in Fig.~\ref{fig:qmdd} is the QMDD representation of the \rblue{controlled-T} gate, where the node labelled with $y_0$ represents the original matrix representation of the \rblue{controlled-T} gate and the two 0 attached to it represent the upper right and bottom left sub-matrices which are all 0-matrices. The two nodes labelled with $y_1$ represent the upper left and bottom right sub-matrices which are,  respectively, the identity matrix and the matrix of the \rblue{T} gate.
\end{example}

\section{Tensor and Tensor Network} \label{sec:tensor}

Before describing our data structure TDD, let us briefly recall the basic idea and notations of tensor networks. 
\subsection{Basic concepts}

A \emph{tensor} is a multidimensional linear map associated with a set of indices. In this paper, we assume that each index takes value in $\{0,1\}$. That is, a tensor with index set $I=\{x_1, \ldots, x_n\}$ is simply a mapping $\phi: \{0,1\}^I \to \mathbb{C}$, where $\mathbb{C}$ is the field of complex numbers. Sometimes, to emphasise the index set, we denote such a tensor by $\phi_{x_1\ldots x_n}$ or $\phi_{\vec{x}}$, and its value on the evaluation $\{x_i \mapsto  a_i, 1\leq i\leq n\}$ by
$\phi_{x_1\ldots x_n}(a_1, \ldots, a_n)$, or simply $\phi_{\vec{x}}(\vec{a})$ or even $\phi(\vec{a})$ when there is no confusion. The number $n$ of the indices of a tensor is called its \emph{rank}. 
Scalars, \rmagenta{2-dimensional} vectors, and \rmagenta{$2\times 2$} matrices are rank 0, rank 1, and rank 2 tensors, respectively. 

 
The most important operation between tensors is \emph{contraction}. The contraction of two tensors is a tensor obtained by summing up over shared indices. Specifically, let $\gamma_{\vec{x}, \vec{z}}$ and $\xi_{\vec{y}, \vec{z}}$ be two tensors which share a common index set $\vec{z}$. Then their contraction is a new tensor $\phi_{\vec{x}, \vec{y}}$ with 
\begin{equation}\label{eq:contdef}
	\phi_{\vec{x}, \vec{y}}(\vec{a}, \vec{b})=\sum_{\vec{c}\in \{0,1\}^{\vec{z}}}{\gamma_{\vec{x}, \vec{z}}(\vec{a}, \vec{c})\cdot \xi_{\vec{y}, \vec{z}}}(\vec{b}, \vec{c}).
\end{equation}
Another useful tensor operation is \emph{slicing}, which corresponds to the cofactor operation of Boolean functions. Let $\phi$ be a tensor with index set $I=\{x, x_1,\ldots,x_n\}$. The slicing of $\phi$ with respect to $x = c$ with $c\in \{0,1\}$ is a tensor $\phi|_{x=c}$ over $I'=\{x_1, \ldots, x_n\}$ given by 
\begin{align}
\phi|_{x=c}(\vec{a}):= \phi(c, \vec{a})
\end{align}
for any $ \vec{a} \in \{0,1\}^n$. We call $\phi|_{x=0}$ and $\phi|_{x=1}$ the negative and positive slicing of $\phi$ with respect to $x$, respectively. 
We say an index $x\in I$ is \emph{essential} for $\phi$ if $\phi|_{x=0}\not=\phi|_{x=1}$.

A \emph{tensor network} is an undirected graph $G=(V, E)$ \rmagenta{with zero or multiple open edges}, where each vertex $v$ in $V$ represents a tensor and each edge a common index associated with the two adjacent tensors. \rmagenta{By contracting connected tensors (i.e., vertices in $V$), with an arbitrary order, we get a rank $m$ tensor, where $m$ is the number of open edges of $G$. This tensor, which is independent of the contraction order, is also called the tensor representation of the tensor network. Interested readers are referred to \cite{markov2008simulating} and \cite{biamonte2019lectures} for more detailed introduction.} 






\subsection{Quantum circuits as tensor networks}
The quantum state of a qubit $x$ with vector representation  $[\alpha_0,\alpha_1]^\intercal$ can be described as a rank 1 tensor $\phi_{x}$, where $\phi_x(0)=\alpha_0$ and $\phi_x(1)=\alpha_1$. Moreover, a single-qubit gate with input qubit $x$ and output qubit $y$ can be represented as a rank 2 tensor $\phi_{xy}$. Note that for tensor representation, we do not distinguish \rmagenta{between} input and output indices, \rmagenta{information about which can be naturally implied} when tensors are interpreted as gates or circuits. For example, the tensor representation of a \Z-gate, with $x$ the input and $y$ the output qubit, is $\phi_{xy}(00)= 1$, $\phi_{xy}(01)=\phi_{xy}(10)=0$, $\phi_{xy}(11)=-1$. Likewise, an $n$-qubit gate is represented as a rank 2$n$ tensor. 

\begin{figure}
\centerline{
\begin{tikzcd}[column sep=0.4cm,row sep=0.4cm]
&\qw  & \gate{T}  &\qw    &\ctrl{1}   &   \qw  & \gate{T}  &\qw&\qw\\
&\qw  & \gate{H}  &\qw    &\gate{X}   &   \qw  & \gate{H}  &\qw&\qw\\
\end{tikzcd}

}
    \caption{A quantum circuit with 2 qubits and 5 gates.}
    \label{exp-for-quantum-circuit}
\end{figure}

A little thought shows that a quantum circuit is naturally a tensor network if we view gates as tensors as above. In such a tensor network, each vertex (tensor) represents a quantum state or a quantum gate and each edge a common index of two adjacent tensors. 
The functionality of any quantum circuit involving $n$ qubits is naturally represented as a tensor of rank $2n$, by contracting all the tensors involved, instead of a $2^{n}\times 2^{n}$ transformation matrix. {This shift of perspective not only decreases our cognitive load, potentially, it will also provide a more concise representation of quantum functionality.} 


\begin{figure}
	\centering
	\subfigure[]{
	\includegraphics[width=0.4\textwidth]{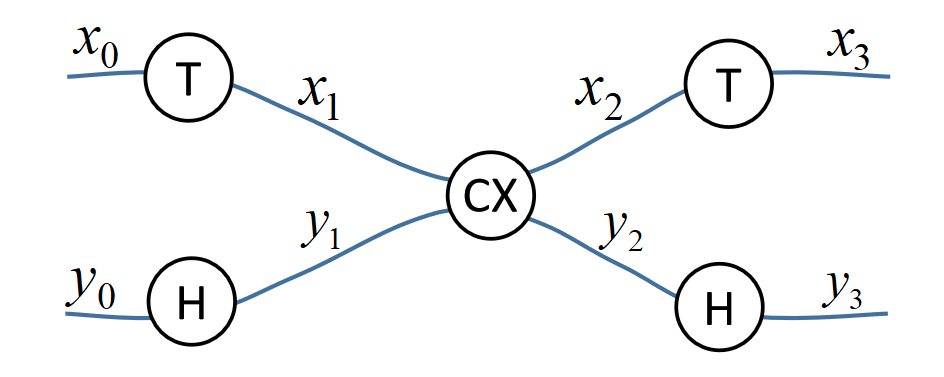}
	}
	\subfigure[]{
	    \includegraphics[width=0.42\textwidth]{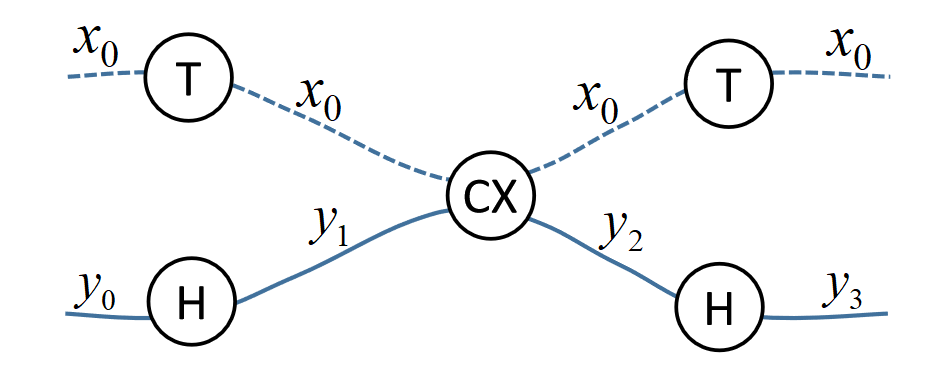}
	 }
	 
	 \caption{Tensor network representations of the  circuit shown in Fig.~\ref{exp-for-quantum-circuit}, where the right one has a hyper-edge (denoted by the dotted line).} \label{exp-for-hyper-edge}
	 
	
\end{figure}

\begin{example}\label{ex:tesnor_of_circuit}
\rmagenta{Consider the circuit shown in Fig.~\ref{exp-for-quantum-circuit}. 
Regarding each gate as a tensor (cf. Fig.~\ref{fig:gates}), Fig.~\ref{exp-for-hyper-edge}(a) shows the tensor network representation of the circuit. By contracting the tensor network, we obtain the tensor representation of the circuit}
%
	\rred{	\begin{align}\label{eq:phi_example}
	\phi_{x_0x_3y_0y_3}(a_0a_3b_0b_3)= \sum_{a_1,a_2,b_1,b_2=0}^{1} \hspace{-5mm} {\T(a_0a_1)\hadamard(b_0b_1)\cx(a_1b_1a_2b_2)\T(a_2a_3)\hadamard(b_2b_3)}.
	\end{align}}
\rmagenta{It is straightforward to check that this tensor 
indeed gives the functionality of the circuit presented in Fig.~\ref{exp-for-quantum-circuit}. For example,   $\phi_{x_0x_3y_0y_3}(1111) = -i$ corresponds to the fact that the circuit maps $\ket{11}$ to $-i\ket{11}$. }
\end{example}


Given a tensor $\phi_{\vec{x}}$ and $x_i, x_j\in \vec{x}$, if $\phi_{\vec{x}}(\vec{a})=0$ whenever $a_i\neq a_j$, we slightly abuse the notation to use an identical index for both $x_i$ and $x_j$. For example, the tenor for \Z\ gate can be written as  $\phi_{xx}$ with $\phi_{xx}(0)=1$ and $\phi_{xx}(1)=-1$. 
Similarly, \cnot\ gate can be represented as a tensor $\phi_{xxy_1y_2}$ with $\phi_{xxy_1y_2}(abc)= a\cdot (b\oplus c) + \overline{a}\cdot \overline{b\oplus c}$, where $\overline{a}$, for example, is the complement of $a$. In \cite{pednault2017breaking}, edges formed by identical indices are called \emph{hyper-edges}.

\begin{example}
For the tensor network shown in Fig.~\ref{exp-for-hyper-edge}(a), the four indices $x_0,x_1,x_2,x_3$ can all be represented by the same index $x_0$ since the two \T\ gates are diagonal and the \cnot\ gate is block diagonal\xblue{, which implies that the tensor value is 0 unless the values of $x_1,x_2,x_3$ are all equal to that of $x_0$.} 
Thus, the tensor network can be modified as the graph shown in Fig.~\ref{exp-for-hyper-edge}(b), where the dotted line is a hyper-edge \rblue{and the corresponding tensor becomes $\phi_{x_0x_0y_0y_3}$}.
\end{example}

\section{Tensor Decision Diagram} \label{sec:tdd}

To fully exploit the benefit of tensor network representation of quantum circuits and the circuit partition technique, a suitable data structure for tensors is desired. In this section, we introduce such a data structure --- Tensor Decision Diagram (TDD).

As decision diagrams, TDDs are similar to ROBDDs \cite{bryant1986graph}, multiplicative binary moment diagrams (*BMDs) \cite{bryant1995verification}, and QMDDs \cite{niemann2015qmdds}, which are designed for representing, respectively, Boolean functions, pseudo–Boolean functions, and $r^{n}\times r^{n}$ matrices. Like $^*$BMDs and QMDDs, TDDs have weights associated with their edges which are combined multiplicatively. In addition, to make the TDD representation of a quantum functionality canonical, several reduction and normalisation rules are also  introduced for TDDs.

While the construction of a QMDD is based on a particular matrix multiplication \cite{niemann2015qmdds}, TDD relies on the Boole-Shannon expansion (see below). For a quantum circuit, nodes in its TDD representation correspond to indices in the circuit (regarded as a tensor network) and each node has two child nodes according the Boole-Shannon expansion. By contrast, nodes in the QMDD representation correspond to qubits in the circuit and each node has four child nodes.

Note that most proofs are deferred to the appendix.


% ...

\subsection{Basic Definition}

To begin with, we observe that any tensor $\phi$ can be expanded with respect to a given index in the style of  \emph{Boole-Shannon expansion} for classical Boolean circuits.
\begin{lemma}\label{lem:shannon-expansion}
Let $\phi$ be a tensor with indices in  $I$. For each $x\in I$, 
\begin{align}\label{eq:shannon}
\phi = \overline{x} \cdot \phi|_{x = 0} + x \cdot \phi|_{x = 1},
\end{align}
where $\overline{x}(c) := 1-x(c)$ for $c\in \{0,1\}$.
\end{lemma}
Note that in above we regard each index $x\in I$ as the identity tensor with only one index $x$, which maps 0 to 0 and 1 to 1.

Recursively using the Boole-Shannon expansion, a tensor can be naturally represented with a decision diagram. 
%
%
%

\begin{definition}[Tensor Decision Diagram]
	A \emph{Tensor Decision Diagram} (TDD) $\F$ over a set of indices $I$ is a rooted, weighted, and directed acyclic graph 
	$\F = (V, E, index, value, low, high, w)$ defined as follows:
	\begin{itemize}
		\item $V$ is a finite set of nodes which is partitioned into non-terminal nodes $V_N$ and terminal ones $V_T$. Denote by $r_\F$ the unique root node of $\F$;
		\item $index: V_N \rightarrow I$ assigns each non-terminal node an index in $I$;
		\item  $value: V_T \rightarrow \mathbb{C}$ assigns each terminal node a complex value;
		\item $low$ and $high$ are mappings in $V_N \rightarrow V$ which assign each non-terminal node with its 0- and 1-successors, respectively;
		\item $E = \{(v, low(v)), (v, high(v)) : v\in V_N\}$ is the set of edges, where $(v, low(v))$ and $(v, high(v)) $ are called the low- and high-edges of $v$, respectively. For simplicity, we also assume the root node $r_\F$ has a unique incoming edge, denoted  $e_r$, which has no source node;
		
		\item $w: E\rightarrow \mathbb{C}$ assigns each edge a complex weight. 		In particular, $w(e_r)$ is called the weight of $\F$, and denoted $w_\F$.
	\end{itemize} 
A TDD is called \emph{trivial} if its root node is also a terminal node.
\end{definition}
For convenience, we often call a terminal node with value $c$ a terminal $c$ node or simply terminal $c$ if it is unique. 

The following example shows how a tensor can be transformed to a TDD using the Boole-Shannon expansion.

\begin{figure}
    \centering
    \includegraphics[width=0.28\textwidth]{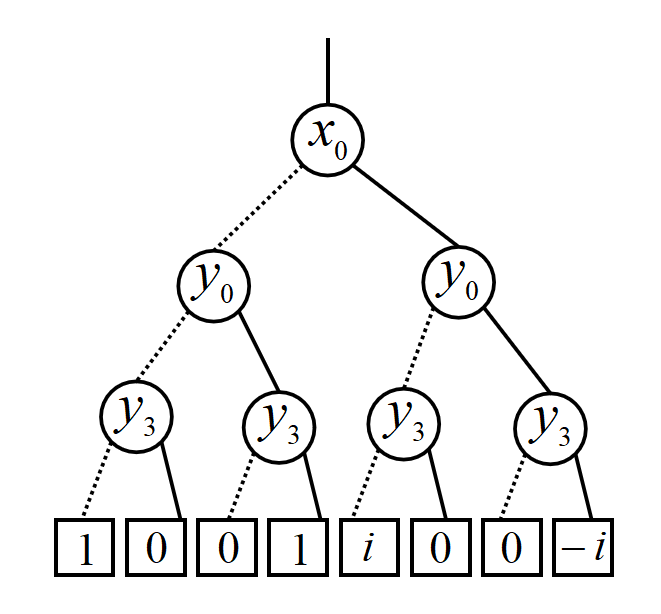}
    \caption{A TDD representation of the tensor in Example 3, where the values of the tensor are stored at the terminal nodes of the TDD.}
    \label{TDD_run_exp1}
\end{figure}

\begin{example}
\rblue{
\rmagenta{Fig.~\ref{TDD_run_exp1} gives the TDD obtained by directly applying the Boole-Shannon expansion to the tensor $\phi_{x_0x_0y_0y_3}$ in Eq.~\ref{eq:phi_example}, where and in all illustrations in this paper we omit the weight of an edge if it is 1. Each terminal node bears a value which, when multiplying with weights along the path to the root node (which happen to be all 1 in this example), corresponds to the value of $\phi$ under the evaluation specified by the path. For example, the terminal node with value $i$ corresponds to the value of $\phi$ under the evaluation $\{x_0\mapsto 1, y_0\mapsto 0, y_3\mapsto 0\}$. Each non-terminal node $v$ acts as a decision node and represents an index $x$, while its low- and high-edges denote evaluations which evaluate $x$ to 0 and, respectively, 1.}}

\end{example}

Conversely, let us see how each node $v$ of a TDD $\F$  naturally corresponds to a tensor $\bphi(v)$.  
If $v$ is a terminal node, then $\bphi(v):=value(v)$ is a rank 0 tensor, i.e., a constant; if $v$ is a non-terminal node, then 
\begin{equation}\label{eq:tdd-intp}
\bphi(v) :=  w_0\cdot \overline{x}_v \cdot \bphi(low(v))
+w_1 \cdot x_v \cdot \bphi(high(v)),
\end{equation}
where $x_v=index(v)$, and $w_0$ and $w_1$ are the weights on the low- and high-edges of $v$, respectively. Comparing Eq.~\ref{eq:tdd-intp} with the Boole-Shannon expansion in Lemma~\ref{lem:shannon-expansion}, we immediately have
\begin{equation}\label{eq:tddnode}
\bphi(v)|_{x_v=c} = w_c\cdot \bphi(v_c),
\end{equation}
where $c\in \{0,1\}$, $v_0 = low(v)$, and $v_1 = high(v)$.

Finally, the tensor represented by $\F$ itself is defined to be 
\begin{align}\label{eq:Phi_of_TDD}
\bphi(\F) := w_\F \cdot \bphi(r_\F).
\end{align}
Recall here that $r_\F$ and $w_\F $ are the root node and the weight of $\F$, respectively.
\rmagenta{An efficient manipulation of general TDDs seems impossible. Following \cite{bryant1986graph}, we restrict our discussion to ordered TDDs.} 
\begin{definition}
A TDD $\mathcal{F}$ is called \emph{ordered} if there is a linear order $\prec$ on $I$ such that $index(v)\prec index(low(v))$ and $index(v)\prec index(high(v))$ for every non-terminal node $v$, provided that both $low(v)$ and $high(v)$ are non-terminal as well. If this is the case, we say $\mathcal{F}$ is a $\prec$-ordered TDD. 

For simplicity, we abuse the notation slightly by assuming $x\prec index(v)$ for all $x\in I$ and all terminal nodes $v\in V_T$. 
\end{definition}

The size of a TDD $\F$, written $|\F|$, is the number of non-terminal nodes of $\F$. As each non-terminal node has two outgoing edges, there are altogether $1+2\times |\F|$ edges, including $e_r$, in  $\F$. Like ROBDDs, the size of the TDD representation strongly relies on the selected variable order. For example, the tensor $\phi = (x_1\cdot x_2)+(x_3 \cdot x_4) + (x_5 \cdot x_6)$ can be represented as a TDD with 6 non-terminal nodes under the order $\prec_1 := (x_1,x_2,x_3,x_4,x_5,x_6)$, but its TDD representation under $\prec_2 := (x_1,x_3,x_5,x_2,x_4,x_6)$ requires at least $2\times (1+2^1+2^2)=14$ non-terminal nodes (cf.~\cite[Ch.3]{molitor2007equivalence}).  While finding an optimal order is NP-hard, there are efficient heuristic methods that have been devised for ROBDDs, which may also be extended to TDDs.


\subsection{Normalisation}
A tensor may have many different TDD representations. For example, let $\F$ be a TDD with root node $r_\F$ and weight $w_\F \not = 0$. A different TDD representing the same tensor can be constructed by, for example, multiplying $w_\F$ by 2 and dividing the weights of the low- and high-edges of $r_\F$ by 2. In order to provide a canonical representation, we first introduce the notion of normal tensors.  
\begin{definition}[normal tensor]\label{dfn:normaltensor}
Let $\phi$ be a tensor with index set $I=\{x_1,\ldots,x_n\}$ and $\prec$ a linear order on $I$. We write 
\begin{align}
\|\phi\| := \max_{\vec{a} \in \{0,1\}^I}|\phi(\vec{a})|
\end{align} 
for the maximum norm of $\phi$. Let $\vec{a}^{*}$ be the first element in $\{0,1\}^I$ (under the lexicographical order induced by $\prec$) which has the maximal 
\rred{magnitude} 
under $\phi$, i.e.,
\begin{align}
\vec{a}^{*}=\min \{\vec{a} \in \{0,1\}^I: |\phi(\vec{a})| = \|\phi\|\}.
\end{align}
We call $\vec{a}^{*}$ the {\emph{pivot}} of $\phi$. A tensor $\phi$ is called {\emph{normal}}  if either $\phi=0$ or $\phi(\vec{a}^*)=1$. 
\end{definition}

It is easy to see that there are tensors $\phi$ with $\|\phi\|=1$ but $\phi$ is not normal. The following lemma shows that any tensor can be normalised in a unique way.

\begin{lemma}\label{lem:tensor-phase}
For any tensor $\phi$ which is not normal, there exists a unique normal tensor $\phi^*$ such that $\phi = p\cdot \phi^*$, where $p$ is a nonzero complex number.
\end{lemma}

The uniqueness of the normal tensor in the above lemma suggests the following definition.
\begin{definition}
A TDD $\F$ is called normal if $\bphi(v)$ is a normal tensor for every node $v$ in $\F$.
\end{definition}
It is worth noting that as normal TDDs may still have arbitrary weights, tensors represented by normal TDDs do not have to be normal. Normal TDDs enjoy some nice properties collected in the following two lemmas.

\begin{lemma}\label{lem:sameNormalTDDs}
Suppose $\F$ and $\g$ are two normal TDDs such that $\bphi(\F)=\bphi(\g)$. Then we have $w_\F= w_\g$ and $\bphi(r_\F)=\bphi(r_\g)$.
\end{lemma}
\begin{proof}
By Eq.~\ref{eq:Phi_of_TDD}, we have $\bphi(\F)=w_\F \cdot \bphi(r_\F)$ and $\bphi(\g)=w_\g \cdot \bphi(r_\g)$. Because $\bphi(r_\F)$ and $\bphi(r_\g)$ are normal tensors and $\bphi(\F)=\bphi(\g)$, by Lemma~\ref{lem:tensor-phase}, we know 
$w_\F=w_\g$ and $\bphi(r_\F)=\bphi(r_\g)$.
\end{proof}

For any non-normal TDD $\F$, we can transform it into a normal one by applying the following two rules.

\vspace{1em}

\noindent\textbf{Normalisation Rules. }  
\begin{itemize}
\item [NR1:] If $v$ is a terminal node with a nonzero value $value(v) \not=1$, then set its value to 1, and change the weight $w$ of each incoming edge of $v$ to $value(v)\cdot w$.

\item [NR2:] 
Suppose $v$ is a non-terminal node such that $\bphi(v)\not=0$ is not normal but both $\bphi(low(v))$ and $\bphi(high(v))$ are normal. Let $w_0$ and $w_1$ be the weights on the low- and high edges of $v$ respectively. If $\bphi(low(v))\neq 0$ and either $\bphi(high(v))=0$ or $|w_0| \geq |w_1|$, we set $w$ to be $w_0$. Otherwise, set it to be $w_1$. Divide $w_0$ and $w_1$ by $w$ and multiply the weight of each incoming edge of $v$ by $w$.

\end{itemize}

Let $\F$ be a non-normal TDD. We first apply NR1 to every terminal node of $\F$ to make it normal. Furthermore, suppose a non-terminal node $v$ of $\F$ represents a non-normal tensor but both its successors represent normal tensors. It is easy to see that, after applying NR2 to $v$, node $v$ itself represents a normal tensor. This gives a procedure to transform $\F$ into a normal TDD in a bottom-up manner. Furthermore, the transformation can be done within time linear in the size of $\F$.

\begin{theorem}
Applying a normalisation rule to a TDD does not change the tensor it represents. Moreover, a TDD is normal if and only if no normalisation rule is applicable.
\end{theorem}
\begin{proof}
The first part of the theorem follows from Eq.~\ref{eq:tdd-intp}, and the second directly from the definitions.
\end{proof}

\begin{figure}
    \centering
    \subfigure[]{
    \includegraphics[width=0.28\textwidth]{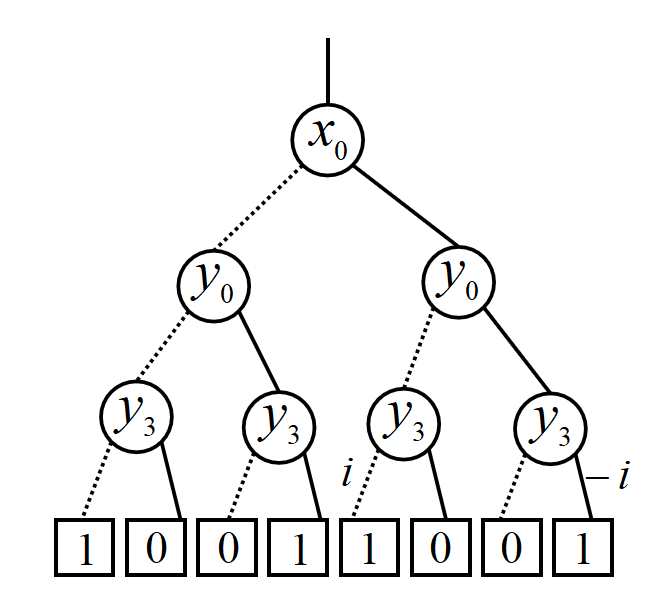}
    }
    \subfigure[]{
    \includegraphics[width=0.28\textwidth]{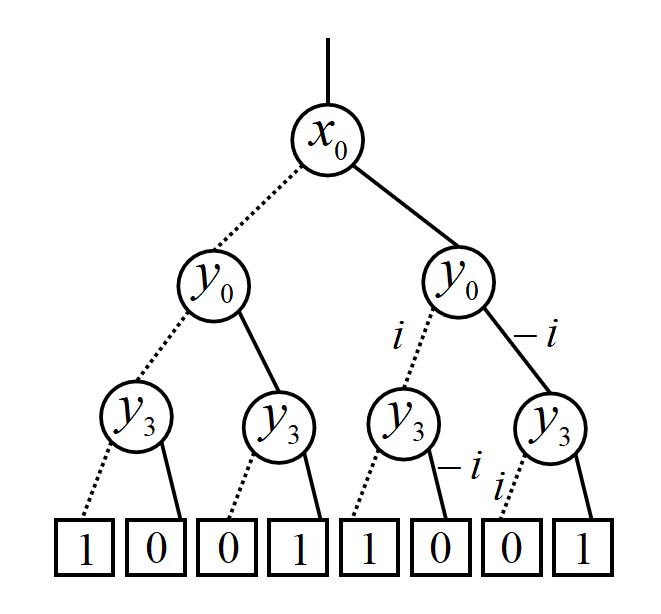}
    }
    \subfigure[]{
    \includegraphics[width=0.28\textwidth]{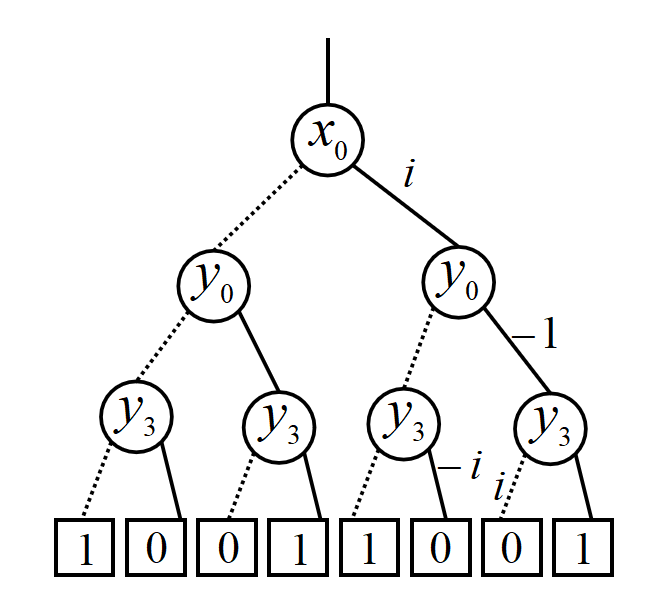}
    }
    \caption{Normalisation of the TDD shown in Fig.~\ref{TDD_run_exp1}, where nodes are normalised, from bottom to top, by applying NR1 or NR2 step by step.
    }
    \label{run_exp_normal}
\end{figure}

\begin{example}\label{ex:normalisation}
\rblue{Applying NR1 to the two terminal nodes labelled with $i$ and $-i$ in the TDD in Fig.~\ref{TDD_run_exp1}, we have the TDD as shown in Fig.~\ref{run_exp_normal}(a). Then, applying NR2 to the right two $y_3$ nodes gives the TDD in Fig.~\ref{run_exp_normal}(b). The normalised TDD, shown in Fig.~\ref{run_exp_normal} (c),  is obtained by applying NR2 to the right $y_0$ node.}

\end{example}

We have seen how to transform an existing TDD into a normal one. In contrast, the following theorem asserts that every tensor has a normal TDD representation. Actually, the proof of this theorem also provides a way to construct a normal TDD directly from a given tensor.



\begin{theorem}\label{thm:complete}
Let $I = \{x_1,x_2,...,x_n\}$ be a set of indices and $\prec$ a linear order on it. For any tensor $\phi$ with index set $I$, there exists a $\prec$-ordered normal TDD $\F$ such that $\bphi(\F) = \phi$. 
\end{theorem}

\subsection{Reduction}
 \rmagenta{As can been seen from Fig.~\ref{run_exp_normal}, normal TDDs  may still have redundant nodes. For example, the first and the third $y_3$ nodes of the normal TDD in Fig.~\ref{run_exp_normal}(c) have the same low- and high-edges and thus represent the same tensor.} This fact motivates us to further introduce: 
\begin{definition}\label{def:recuded}
A TDD $\F$ is called \emph{reduced} if it is normal and
\begin{enumerate}
\item no node represents the 0 tensor, i.e., $\bphi(v)\not= 0$ for every node $v$ in $\F$; 
\item all edges weighted 0 point to the (unique) terminal 1; 
and 
\item no two different nodes represent the same tensor, i.e., $\bphi(u)\not= \bphi(v)$ for any two nodes $u\not=v$ in $\F$.
\end{enumerate} 
\end{definition}

The following lemma shows that every non-terminal node of a reduced TDD $\F$ is labelled with an essential variable of the tensor represented by $\F$.

\begin{lemma}\label{lem: nodes_in_reducedTDD_are_essential}
Suppose $\F$ is a reduced TDD of a non-constant tensor $\phi$ over index set $I$. Then every non-terminal node of $\F$ is labelled with an index that is essential to $\phi$.  
\end{lemma}

The following definition of sub-TDDs is useful in our later discussion. Recall that we assume $x\prec index(v)$ for all $x\in I$ and all terminal nodes $v$.
\begin{definition}\label{def:subtdd}
	Let $\F$ be a reduced TDD over a $\prec$-linearly ordered index set $I$.  Let $x\in I$, and $x\preceq index(r_\F)$. We define sub-TDDs $\F_{x=0}$ and $\F_{x=1}$ of $\F$ as follows.
	\begin{enumerate}
		\item If $x \prec  index(r_\F)$, then $\F_{x=0} =\F_{x=1} = \F$;
		\item If $x = index(r_\F)$, $\F_{x=0}$ is defined as the TDD rooted at $low(r_\F)$ with weight  $w_\F \cdot w(r_\F, low(r_\F))$, i.e., the weight of the low-edge of $r_\F$ multiplied by the weight of $\F$. Analogously, we have $\F_{x=1}$.
	\end{enumerate}
\end{definition}
Corresponding to the Boolean-Shannon expansion for tensors (cf. Eq.~\ref{eq:shannon}), we have
\begin{lemma}\label{lem:tdd-shannon}
Suppose $\F$ is a reduced TDD on $I$, $x\in I$ and $x\preceq index(r_\F)$. Then we have 
\begin{align}\label{eq:tdd-shannon}
\bphi(\F) = \overline{x} \cdot \bphi(\F_{x=0}) + x \cdot \bphi(\F_{x=1}).
\end{align}
\end{lemma}

Now we are ready to prove the canonicity of reduced TDDs. 
Two TDDs $\F$ and $\g$ are said to be isomorphic, denoted $\F\eqsim \g$, if they are equal up to renaming of the nodes; that is, there exists a graph isomorphism between $\F$ and $\g$ which preserves node indices, edge weights, and values on terminal nodes. Furthermore, it maps low-edges to low-edges and high-edges to high-edges.

\begin{theorem}[canonicity]\label{thm:canonicity}
Let $I$ be an index set and $\prec$ a linear order on $I$. Suppose $\F$ and $\g$ are two  $\prec$-ordered, reduced TDDs over $I$ with $\bphi(\F)=\bphi(\g)$. Then $\F\eqsim\g$. 
\end{theorem}



A reduced TDD can be obtained by applying the following reduction rules on any normal TDD in a bottom-up manner. 

\vspace{1em}

\noindent\textbf{Reduction rules. } 
\begin{itemize}
\item [RR1:] 
Merge all terminal 1 nodes.
Delete all terminal 0 ones, if exist, and redirect their incoming edges to the (unique) terminal and reset their weights to 0.  

\item [RR2:] Redirect all weight-0 edges to the terminal. 
If these include the incoming edge of the root node, then the terminal becomes the new root. Delete all nodes (as well as all edges involving them) which are not reachable from the root node. 

\item[RR3:] Delete a node $v$ if its 0- and 1-successors are identical and its low- and high-edges have the same weight $w$ (either $0$ or $1$). Meanwhile, redirect its incoming edges to \rmagenta{terminal 1 if $w=0$ and, if otherwise, to its successor.}


\item [RR4:] Merge two nodes if they have  the same index, the same 0- and 1-successors, and the same weights on the corresponding edges.
\end{itemize}

\begin{theorem}\label{thm:reduced_if_no_rule_applicable}
A normal TDD is reduced if and only if no reduction rule is applicable. 
\end{theorem}

The following theorem guarantees that the reduced TDD of a tensor can be obtained by applying the reduction rules. 
\begin{theorem}
Let $\F$ be a normal TDD representing tensor $\phi$. Applying a reduction rule to $\F$ does not change the tensor it represents. Moreover, the reduced TDD of $\phi$ can be obtained from $\F$ by applying the reduction rules till no one is applicable. 
\end{theorem}
\begin{proof}
It is routine to show that applying any reduction rule to a normal TDD does not change the tensor it represents. Suppose $\F$ is a normal TDD that is not reduced. Applying the reduction rules in a bottom-up manner until no rule is applicable, by Theorem~\ref{thm:reduced_if_no_rule_applicable}, we obtain a reduced TDD that also represents $\phi=\bphi(\F)$.  As reduced TDDs are unique \rmagenta{(see Theorem~\ref{thm:canonicity})}, this gives the reduced TDD of $\phi$.
\end{proof}

As each application of a reduction rule removes some nodes, the reduced TDD has the minimal number of nodes.
\begin{corollary}
Let $\F$ be a normal TDD of a tensor $\phi$. Then $\F$ is reduced if and only if $|\F| \leq |\g|$ for any other normal TDD $\g$ of $\phi$.
\end{corollary}


\begin{figure}
    \centering
    \subfigure[]{
     \centering
    \includegraphics[width=0.28\textwidth]{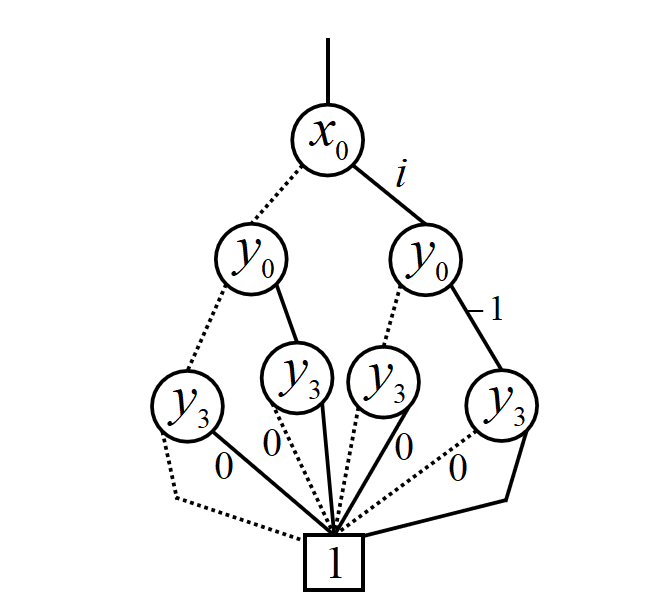}
    }
    \subfigure[]{
    \centering
    \includegraphics[width=0.28\textwidth]{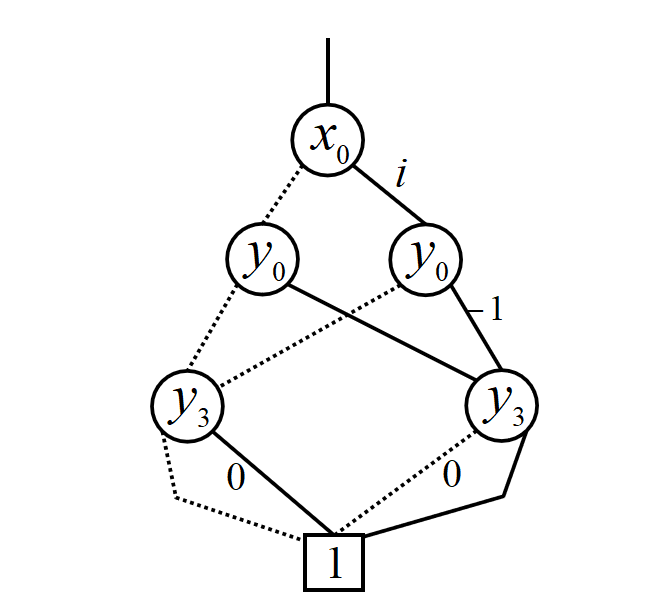}
    }
    \caption{Reduction of the normalised TDD shown in Fig.~\ref{run_exp_normal}(c), where nodes that represent the same tensor (the first and the third $y_3$ nodes, the second and the fourth $y_3$ nodes) in (a) are merged.
    }
    \label{run_exp_reduce}
\end{figure}

\begin{example}\label{ex:reduction}


\rblue{Consider the normalised TDD shown in Fig.~\ref{run_exp_normal}(c). Applying RR1 to merge all terminal 1 nodes and delete all terminal 0 nodes gives the TDD shown in Fig.~\ref{run_exp_reduce}(a). Then, further applying RR4 to merge the first and the third as well as the second and the fourth $y_3$ nodes, we have the reduced TDD as shown in Fig.~\ref{run_exp_reduce}(b)\rmagenta{, which provides a compact representation for the circuit in Fig.~\ref{exp-for-quantum-circuit}}.}

\end{example}

\begin{remark}\label{remark:robdd}
\rmagenta{As Boolean functions are special tensors, each Boolean function also has a unique reduced TDD representation, which can be obtained by  performing the reduction rule RR1 on its ROBDD representation if we assign weight 1 to each ROBDD edge. }
\end{remark}

\subsection{\xblue{Comparison with QMDD}}\label{sec:qmdd_vs._tdd}

As decision diagrams, TDDs are closely related to QMDDs. In this subsection, we provide a detailed investigation in their connection and difference. 

First of all, we observe that QMDDs can be transformed into TDDs in a natural and flexible way. Indeed, 
suppose $\M$ is the QMDD representation of a $2^n\times 2^n$ matrix $A$ w.r.t. variable order $q_1\prec q_2 \prec \cdots \prec q_n$. Recall that each QMDD variable $q_i$ corresponds to two TDD variables, i.e.,  an input index $x_i$ and an output index $y_i$. We say an order $\prec'$ on $I:=\{x_i,y_i\mid 1\leq i\leq n\}$ respects $\prec$ if \blue{(i)} for each $1\leq i\leq n$,  $x_i$ is either the $\prec'$-predecessor or the $\prec'$-successor of $y_i$, \blue{and (ii) for any $1\leq j<i\leq n$, $x_i$ and $y_i$ are $\prec'$-descendants of both $x_j$ and $y_j$.}


\begin{theorem}
Suppose $\M$ is the QMDD representation of a $2^n\times 2^n$ matrix $A$ w.r.t. variable order $q_1\prec q_2 \prec \cdots \prec q_n$. Let $x_i$ and $y_i$ be the input and output indices of $q_i$ for each $1\leq i\leq n$ and set $I:=\{x_i,y_i\mid 1\leq i\leq n\}$. Let $\prec'$ be any order on $I$ that respects $\prec$. Then the TDD representation of $A$ w.r.t. $\prec'$ has $|\M|$ to $3|\M|$ non-terminal nodes, where $|\M|$ is the number of non-terminal nodes of $\M$.
\end{theorem}
\begin{proof}[Proof (Sketch).]

We prove this by using induction on $n$, the number of variables in $\M$. 

We start with $n=1$. Let $A=\begin{bmatrix}a & b \\ c & d\end{bmatrix}$ be an arbitrary $2\times 2$ matrix with variable $q_1$, and input/output indices $x_1$ and $y_1$. From Table~\ref{tab:qmdd_vs_tdd}, it is easy to see that, if $A$ is non-trivial (i.e., not all entries are identical), then the QMDD has only one non-terminal node, while 
the TDD w.r.t. $x_1\prec y_1$ has $1\leq \ell \leq 3$ non-terminal nodes. This conclusion also applies to the TDD w.r.t. $y_1\prec x_1$, which may be different from the TDD w.r.t. $x_1\prec y_1$. 

In general, let $A=\begin{bmatrix}A_{00} & A_{01} \\ A_{10} & A_{11}\end{bmatrix}$ be a $2^s\times 2^s$ matrix, where $A_{ij}$ are $2^{s-1}\times 2^{s-1}$ matrices for $0\leq i,j\leq 1$. W.l.o.g., we assume that these $A_{ij}$ are not the same. Suppose $\M_{ij}$ and $\cX{T}_{ij}$ are, respectively, the QMDD and TDD representation for $A_{ij}$, w.r.t. the corresponding restricted order. Then from the induction hypothesis, we know
 $$|\M_{ij}|\leq |\cX{T}_{ij}| \leq 3|\M_{ij}|.$$ 
Furthermore, for each non-terminal node $nd$ in $\M_{ij}$, there exists an non-terminal node $nd'$ in $\cX{T}_{ij}$ such that the QMDD rooted at $nd$, denoted by $\M_{nd}$, represents the same matrix as the TDD rooted at $nd'$, denoted by $\cX{T}_{nd'}$. Clearly, the above inequality also holds for $\M_{nd}$ and $\cX{T}_{nd'}$. While $\M_{ij}$ may share nodes with some other $\M_{i'j'}$, it can be further proved (omitted here) that 
\begin{align}\label{eq:T3M}
\bigg|\bigcup_{0\leq i,j\leq 1}{nodes(\M_{ij})}\bigg| \leq \bigg|\bigcup_{0\leq i,j\leq 1}{nodes(\cX{T}_{ij})}\bigg| \leq 3 \bigg|\bigcup_{0\leq i,j\leq 1}{nodes(\M_{ij})}\bigg|,
\end{align}
where, for a QMDD or TDD $\cX{X}$, $nodes(\cX{X})$ denotes  the set of non-terminal nodes in $\cX{X}$.

Suppose the variable order taken by the QMDD representation $\M$ of $A$ is $q_1\prec q_2\prec \cdots\prec q_n$. Then $q_1$ is the root node of $\M$ and it has at most four child nodes, which are the root nodes of $\M_{ij}$ for  $0\leq i,j\leq 1$. Thus, 
$|\M| = 1 + |\bigcup_{0\leq i,j\leq 1}{nodes(\M_{ij})}|$.
Let $\prec'$ be an order on $I$ which respects $\prec$. Then $x_1$ and $y_1$ must be the two least indices in the order. Let us assume $x_1\prec' y_1$.  We construct a TDD representation $\cX{T}$ for $A$ w.r.t. $\prec'$. Let $x_1$ be the root of $\cX{T}$, and label with $y_1$ the two child nodes of $x_1$. The left (right, resp.) $y_1$ node has two child nodes which are, respectively, the root nodes of $\cX{T}_{00}$ and $\cX{T}_{01}$ ($\cX{T}_{10}$ and $\cX{T}_{11}$, resp.). Assume that the edges from these $y_1$ nodes are properly attached with the weights of the corresponding $\cX{T}_{ij}$. Then $\cX{T}$ is a (possibly non-reduced) TDD representation of $A$ and we have $|\cX{T}| = 3 +  |\bigcup_{0\leq i,j\leq 1}{nodes(\cX{T}_{ij})}|$. Note that here $\cX{T}$ may be further reduced by merging root nodes of $\cX{T}_{ij}$ and $\cX{T}_{i'j'}$ if they are identical, but all non-terminal nodes below them are non-redundant. Also write $\cX{T}$ for the reduced TDD. Then $\cX{T}$ has one to three more non-terminal nodes than $\bigcup_{0\leq i,j\leq 1}{nodes(\cX{T}_{ij})}$. Thus, by Eq.~\ref{eq:T3M} and $|\M| = 1 + |\bigcup_{0\leq i,j\leq 1}{nodes(\M_{ij})}|$, we have  $|\M| \leq |\cX{T}| \leq 3|\M|$.
\end{proof}

\begin{table}[]
\centering
\caption{Representation of $2\times 2$ matrix: QMDD vs. TDD, where $a,b,c,d$ are different complex values other than 0, $k$ is a nonzero complex value. The rows show, respectively, the matrix $A$, the QMDD representation $\M$, the TDD representation $\cX{T}_1$ with $x_1\prec y_1$, the TDD representation $\cX{T}_2$ with $y_1\prec x_1$, where for illustrative purpose we take $a,b,c,d,k$ as 1,0.8,0.7,0.5,0.5, respectively. Note that, except those cases when 0 is an entry in $A$, the cases when $A_{10}=A_{11}$
are omitted, as they are symmetrical to the cases when $A_{00}=A_{01}$.
} \label{tab:qmdd_vs_tdd}
\scalebox{0.85}{
\begin{tabular}{c|c|c|c|c|c|c|c}
    $\begin{bmatrix}a & a \\ a & a\end{bmatrix}$
    &
    $\begin{bmatrix}a & a \\ b & b\end{bmatrix}$
    &
    $\begin{bmatrix}a & b \\ a & b\end{bmatrix}$
    &
    $\begin{bmatrix}a & ka \\ b & kb\end{bmatrix}$
    &
    $\begin{bmatrix}a & b \\ ka & kb\end{bmatrix}$
    &
    $\begin{bmatrix}a & a \\ b & c\end{bmatrix}$
    &
    $\begin{bmatrix}a & b \\ a & c\end{bmatrix}$
    &
    $\begin{bmatrix}a & b \\ c & d\end{bmatrix}$
    \\
    $\raisebox{-.5\height}{\includegraphics[width=0.08\linewidth]{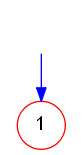}}$
    &$\raisebox{-.5\height}{\includegraphics[width=0.12\linewidth]{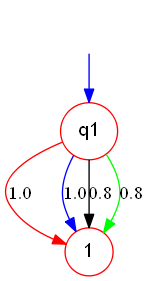}}$    
    &$\raisebox{-.5\height}{\includegraphics[width=0.12\linewidth]{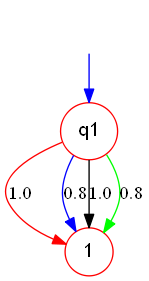}}$
    &
    $\raisebox{-.5\height}{\includegraphics[width=0.12\linewidth]{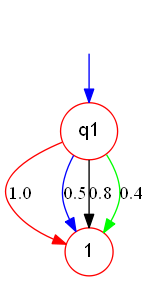}}$    
    &
    $\raisebox{-.5\height}{\includegraphics[width=0.12\linewidth]{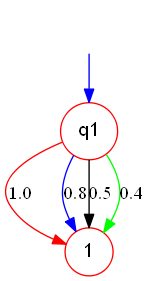}}$ 
    &
    $\raisebox{-.5\height}{\includegraphics[width=0.12\linewidth]{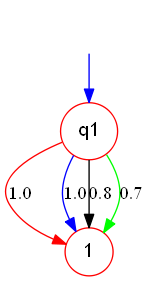}}$
    &
    $\raisebox{-.5\height}{\includegraphics[width=0.12\linewidth]{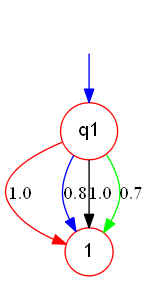}}$
    & 
    $\raisebox{-.5\height}{\includegraphics[width=0.12\linewidth]{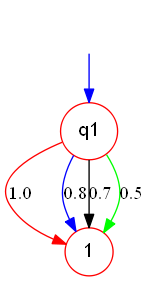}}$   
    \\
    $\raisebox{-.5\height}{\includegraphics[width=0.08\linewidth]{figure_m/1.png}}$   
    &
    $\raisebox{-.5\height}{\includegraphics[width=0.07\linewidth]{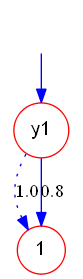}}$    
    &
    $\raisebox{-.5\height}{\includegraphics[width=0.07\linewidth]{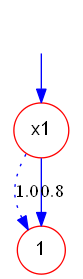}}$
    &
    $\raisebox{-.5\height}{\includegraphics[width=0.07\linewidth]{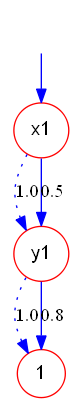}}$
    &
    $\raisebox{-.5\height}{\includegraphics[width=0.07\linewidth]{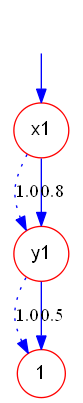}}$   
    &
    $\raisebox{-.5\height}{\includegraphics[width=0.135\linewidth]{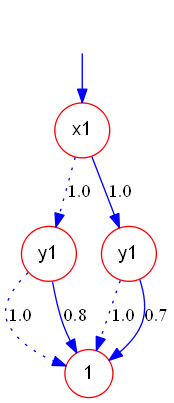}}$   
    &
    $\raisebox{-.5\height}{\includegraphics[width=0.105\linewidth]{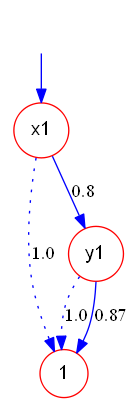}}$
    &
    $\raisebox{-.5\height}{\includegraphics[width=0.14\linewidth]{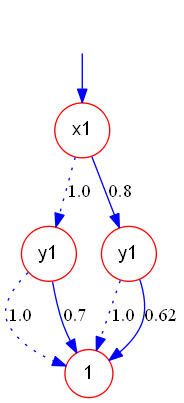}}$ 
    \\
    $\raisebox{-.5\height}{\includegraphics[width=0.08\linewidth]{figure_m/1.png}}$  
    &
    $\raisebox{-.5\height}{\includegraphics[width=0.07\linewidth]{figure_m/2_2.png}}$
    &
    $\raisebox{-.5\height}{\includegraphics[width=0.07\linewidth]{figure_m/3_2.png}}$
    &
    $\raisebox{-.5\height}{\includegraphics[width=0.07\linewidth]{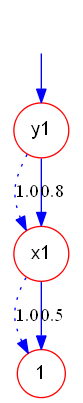}}$
    &
    $\raisebox{-.5\height}{\includegraphics[width=0.07\linewidth]{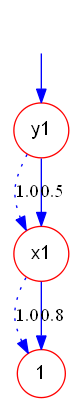}}$
    &
    $\raisebox{-.5\height}{\includegraphics[width=0.105\linewidth]{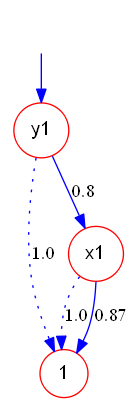}}$
    &
    $\raisebox{-.5\height}{\includegraphics[width=0.135\linewidth]{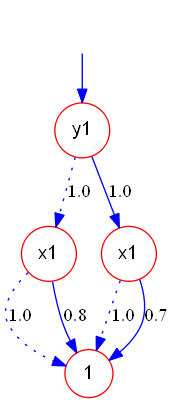}}$
    &
    $\raisebox{-.5\height}{\includegraphics[width=0.14\linewidth]{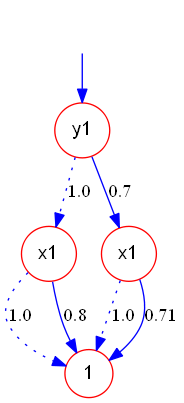}}$
\end{tabular}
}
\end{table}




From the above proof we can see that each non-terminal node in the QMDD representation is transformed into 1 to 3 non-terminal nodes in the TDD representation w.r.t. to any index order that respects the variable order. The proof indeed specifies a linear-time procedure for transforming a given QMDD to a TDD. 
It is worth noting that not all TDDs can be derived in this way and TDD can be more flexible as it is defined for arbitrary tensors (i.e., not necessarily $2^n\times 2^n$ matrices) and with respect to arbitrary index orders.

\begin{figure}
\centering
\subfigure[]{
\centering
\begin{tikzpicture}
\node at(0,0) [anchor=north]{$\begin{bmatrix}1 & 0 &0 &0 \\ 0 & 0 &1 &0 \\ 0&1&0&0 \\0&0&0&1\end{bmatrix}$};
\end{tikzpicture}
}
\subfigure[]{
\centering
\begin{tikzpicture}
\node at (0,0) [anchor=north]{
\begin{tikzcd}[column sep=0.8cm,row sep=0.3cm]
\lstick{$q_1:\ket{i}$}  &\swap{2} &\qw\rstick{$\ket{j}$}& &\lstick{$\ket{i}$} &\qw&\qw\rstick{$\ket{i}$}        \\
                 & & &=      \\
\lstick{$q_2:\ket{j}$}   &\targX{} &\qw\rstick{$\ket{i}$}& &\lstick{$\ket{j}$}&\qw&\qw\rstick{$\ket{j}$}   \\
\end{tikzcd}};
\node at(-2,0.1) [anchor=north]{$x_{1}$};
\node at(-0.8,0.1) [anchor=north]{$y_{1}$};
\node at(-2,-1) [anchor=north]{$x_{2}$};
\node at(-0.8,-1) [anchor=north]{$y_{2}$};
\node at(1.5,0.1) [anchor=north]{$x_{1}$};
\node at(2.7,0.1) [anchor=north]{$y_{2}$};
\node at(1.5,-1) [anchor=north]{$x_{2}$};
\node at(2.7,-1) [anchor=north]{$y_{1}$};
\end{tikzpicture}
}
\caption{The SWAP gate: (a) matrix representation, and (b) two equivalent tensor network representations. 
}
\label{fig:swap-tensor}
\end{figure}

\begin{figure}
    \centering
    \subfigure[]{
    \centering
    \includegraphics[width=0.25\textwidth]{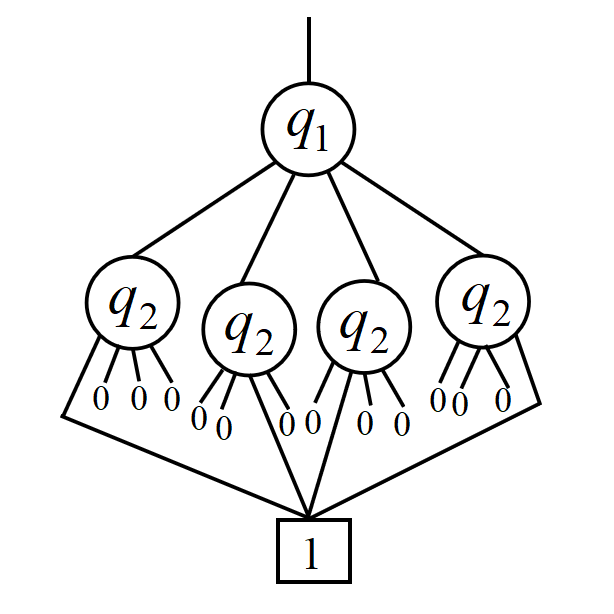}
    }
    \subfigure[]{
    \centering
    \includegraphics[width=0.2\textwidth]{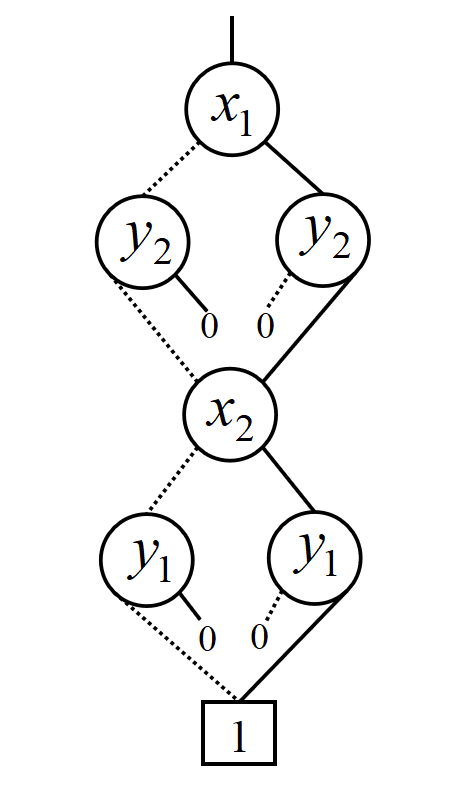}
    }    
    \caption{Decision diagrams for the SWAP gate: (a) QMDD w.r.t. $q_1\prec q_2$, (b) TDD w.r.t $x_1\prec' y_2 \prec' x_2 \prec' y_1$.}
    \label{fig:swap_dd}
\end{figure}

\begin{example}\label{ex:SWAP}
Consider the SWAP gate between qubits $q_1$ and $q_2$  (see Fig.~\ref{fig:swap-tensor}), which swaps the state on $q_1$ with that on $q_2$ (i.e.,   $SWAP\ket{i}\ket{j}=\ket{j}\ket{i}$ for any $i,j\in\{0,1\}$). The QMDD representation w.r.t. order $q_1\prec q_2$ of the SWAP gate has 6 nodes  (cf.  Fig.~\ref{fig:swap_dd}(a)), and the TDD representation w.r.t. any order that respects $q_1\prec q_2$ has 10 nodes. However, let $x_i$ and $y_i$ be the input and output indices corresponding to $q_i$ for $i=1,2$. By putting $x_1$ next to $y_2$ and $x_2$ next to $y_1$, we have a TDD representation with only 7 nodes (cf. Fig.~\ref{fig:swap_dd}(b)).
\end{example}


The above discussion shows that, given the QMDD representation $\M$ of a $2^n\times 2^n$ matrix $A$, we can construct from $\M$,  in linear time, a TDD representation of $A$ with size $\leq 3|\M|$. This suggests that TDD, as a data structure, is universal and as compact as QMDD. When the task is to represent the functionality of a quantum circuit, however, TDD could be more flexible as the ordering of non-terminal indices in the circuit need not to respect the qubit order. This can be compared with QMDD: once the qubit order has been selected, usually we represent the unitary matrix of each gate in the circuit w.r.t. the selected order. By exploiting the flexibility of ordering among non-terminal indices,  sometimes we can have an exponential reduction in the size of the decision diagrams generated during the computing process.

To describe such an example, we need the following simple observation.
\begin{example}\label{ex:C_sigma}
Let $C_{id}$ be the SWAP circuit with $2n$ qubits in $Q:=\{q_1,\ldots,q_{2n}\}$ which swaps $q_k$ and $q_{n+k}$ for $1\leq k \leq n$. Then its QMDD w.r.t. qubit order $q_1\prec q_2\prec \cdots\prec q_{2n}$ has $\frac{1}{3}(5\cdot 4^{n}-2)$ nodes.\footnote{This is because different values of the first $n$ qubits have different successor values on the last $n$ qubits.} In general, for any permutation $\sigma$ on $Q$, let $C_{\sigma}$ be the quantum circuit which swaps $q_{\sigma(k)}$ and $q_{\sigma(n+k)}$ for $1\leq k \leq n$. We also  need $\frac{1}{3}(5\cdot 4^{n}-2)$ nodes to represent $C_{\sigma}$ as a QMDD w.r.t. qubit order $q_{\sigma(1)}\prec q_{\sigma(2)}\prec \cdots \prec q_{\sigma(2n)}$.



On the other hand, let $x_k$ and $y_k$ be the input and output indices corresponding to $q_k$. The TDD w.r.t. index order $x_{\sigma(1)}\prec' y_{\sigma(n+1)} \prec' \cdots \prec' x_{\sigma(n)} \prec' y_{\sigma(2n)} \prec' \cdots \prec' x_{\sigma(2n)} \prec' y_{\sigma(n)}$ has size $6n+1$. 

As a simple experiment, we construct the QMDD and TDD for the SWAP circuit $C_{id}$ on $Q:=\set{q_1,\ldots,q_{2n}}$. The results for $1\leq n\leq 10$ are shown in Fig.~\ref{fig:swap_pic}, where the qubit order for QMDD is $q_1 \prec q_2 \prec \cdots \prec q_{2n}$ and the index order for TDD is $x_1 \prec y_{n+1} \prec x_2 \prec y_{n+2} \prec \cdots \prec x_{n+1} \prec y_{1}\prec \cdots \prec x_{2n} \prec y_{n}$.

\begin{figure}
    \centering
    \includegraphics[width=0.55\textwidth]{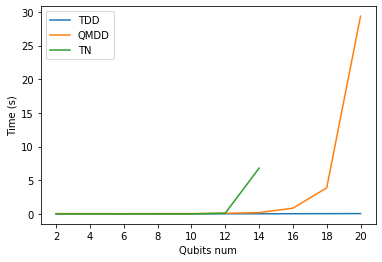}
    \caption{The time consumption for constructing the functionality of the SWAP circuit $C_{id}$ on $2n$ qubits. 
    }
    \label{fig:swap_pic}
\end{figure}
\end{example}


To be fair, if we select the qubit order $q_{\sigma(1)} \prec q_{\sigma(n+1)} \prec q_{\sigma(2)} \prec q_{\sigma(n+2)} \prec \cdots \prec q_{\sigma(n)} \prec q_{\sigma(2n)}$, the above circuit $C_\sigma$ also has a QMDD representation with linear size. The following example shows that, however, there is a circuit $C$ such that, no matter what qubit order is selected, at some step, the QMDD representation of the circuit will have size that is exponential in $n$ if we compute the functionality of the circuit gate by gate according to its topological order.



\begin{example}
We first observe that any SWAP circuit on $Q:=\{q_1,\ldots,q_{2n}\}$, i.e., a circuit consisting only of SWAP gates, is equivalent to $C_\sigma$ for some permutation $\sigma$ on $Q$. As the number of permutations on $Q$ is finite, we can construct a SWAP circuit $C=(g_1,\ldots,g_\ell)$ which satisfies the following condition: for any permutation $\sigma$, there is some $1\leq i\leq \ell$ such that  $C_\sigma \equiv C_i:=(g_1,\ldots,g_i)$. 

Note that each qubit order $\prec$ induces a permutation $\sigma_\prec$ such that  $q_{\sigma_\prec(1)}\prec q_{\sigma_\prec(2)} \prec \cdots\prec q_{\sigma_\prec(2n)}$.
As a consequence, for any qubit order $\prec$, there exists some $1\leq i\leq \ell$ such that $ C_{\sigma_\prec}\equiv C_i$. By the analysis given in Example~\ref{ex:C_sigma}, at this step, the QMDD of $C_i$ has size that is exponential in $n$.

%
Such an example for $n=2$ is illustrated in  Fig.~\ref{exp-for-swap-circuit}.
Note that, after the first two SWAP gates, $q_1$ and $q_3$ are swapped with $q_2$ and $q_4$, respectively; after the first four SWAP gates, $q_1$ and $q_2$ are swapped with $q_3$ and  $q_4$, respectively; in the end of the circuit, $q_1$ and $q_3$ are swapped with $q_4$ and $q_3$, respectively.
%




On the other hand, let us select an index order for TDD as follows: for each SWAP gate $g$ between $q_i$ and $q_j$, let $x_i$, $x_j$ be the input indices of $q_i$ and $q_j$ at $g$, and $y_i$, $y_j$ the output indices of $q_i$ and $q_j$ at $g$ (cf. Fig.~\ref{fig:swap-tensor}(b) for the case when $g$ swaps $q_1$ and $q_2$). Requiring that $x_i$ is adjacent to $y_j$ and $x_j$ is adjacent to $y_i$, it can be proved that the TDD w.r.t. any such an order can be constructed in a way such that any TDD generated in the process has at most $6n+1$ non-terminal nodes. This is because the values of $x_i$ and $y_j$ are always equal, and the same applies to $x_j$ and $y_i$. Thus, it behaves like the tensor product of two identity gates if we see $x_i$ and $y_j$ as the input and output indices of one qubit and $x_j$, $y_i$ of the other (cf. Fig.~\ref{fig:swap-tensor}(b)), which can be represented as a TDD with $2\times 3+1$ nodes. When more SWAP gates are involved, it can also be regarded as a series of tensor products of identity matrices, taking only $6n+1$ nodes.

Consider the circuit given in Fig.~\ref{exp-for-swap-circuit} again.
For this circuit, no matter what qubit order is used, it generates, at some point of the computation, a QMDD with $\frac{1}{3}(5\cdot 4^{2}-2)=26$ nodes.
However, if we set the index order as \[x_1 \prec x_{2,1}\prec x_{3,2} \prec y_4 \prec x_{2} \prec x_{1,1} \prec x_{4,2} \prec y_3 \prec x_3 \prec x_{4,1} \prec x_{1,2} \prec y_2 \prec x_{4} \prec x_{3,1} \prec x_{2,2} \prec y_1,\] 
all TDDs generated in the process have at most 13 (which is $6n+1$ for $n=2$) nodes. It is worth noting that, when $n=10$, the two numbers are 1,747,626 vs. 61.

\end{example}


\begin{figure}
\centerline{
\begin{tikzpicture}
\node at (0,0) [anchor=north]{
\begin{tikzcd}[column sep=0.8cm,row sep=0.4cm]
\lstick{$q_1$}   &\swap{1} &\qw      &\swap{3}  &\swap{1}      &\qw\\
\lstick{$q_2$}   &\targX{} &\swap{1} &\qw       &\targX{}      &\qw\\
\lstick{$q_3$}   &\swap{1} &\targX{} &\qw       &\swap{1}      &\qw\\
\lstick{$q_4$}   &\targX{} &\qw      &\targX{}  &\targX{}      &\qw\\
\end{tikzcd}};
\node at(-2,0.1) [anchor=north]{$x_{1}$};
\node at(-0.8,0.1) [anchor=north]{$x_{1,1}$};
\node at(1.4,0.1) [anchor=north]{$x_{1,2}$};
\node at(2.6,0.1) [anchor=north]{$y_{1}$};
\node at(-2,-0.7) [anchor=north]{$x_{2}$};
\node at(-0.8,-0.7) [anchor=north]{$x_{2,1}$};
\node at(0.3,-0.7) [anchor=north]{$x_{2,2}$};
\node at(2.6,-0.7) [anchor=north]{$y_{2}$};
\node at(-2,-1.4) [anchor=north]{$x_{3}$};
\node at(-0.8,-1.4) [anchor=north]{$x_{3,1}$};
\node at(0.3,-1.4) [anchor=north]{$x_{3,2}$};
\node at(2.6,-1.4) [anchor=north]{$y_{3}$};
\node at(-2,-2.2) [anchor=north]{$x_{4}$};
\node at(-0.8,-2.2) [anchor=north]{$x_{4,1}$};
\node at(1.4,-2.2) [anchor=north]{$x_{4,2}$};
\node at(2.6,-2.2) [anchor=north]{$y_{4}$};
\end{tikzpicture}
}
\caption{A SWAP circuit over four qubits.}
\label{exp-for-swap-circuit}
\end{figure}


While the above circuit construction seems  artificial, SWAP gates are frequently used in practice. For example, they are used at the end of QFT circuits to reverse the order of qubits. They are also used in tasks like quantum circuit transformation (see, e.g., \cite{MaslovFM08,LiZF21}) in order to bring qubits close to each other so that desired 2-qubit operations can be executed on a near-term quantum device like Google's Sycamore. In Section~\ref{sec:circuitpartition} we will further demonstrate the flexibility of TDD calculations by employing circuit partition schemes.

\section{Algorithms} \label{sec:construction}


This section is devoted to algorithms for constructing the corresponding reduced TDD from a given tensor and key operations such as addition and contraction of TDDs. All of these algorithms are implemented in a recursive manner. Every time a new node is generated, we apply normalisation and reduction rules locally to this node, implemented by calling the $reduce$ procedure. In this way, it can be guaranteed that the TDDs obtained are all reduced. It is also worth noting that motivated by~\cite{brace1990efficient}, to avoid redundancy, in our real implementation (not shown in the algorithms) all the nodes are stored in a hash table. Whenever a new node is about to be generated, we first check if such a node (with the same index, successors and weights on the corresponding edges) already exists in the table. If yes, the node is returned directly; otherwise, a new one is created and added into the hash table. 

\subsection{Generation}

 Algorithm~\ref{TDD_generate} shows the process of generating the reduced TDD of a tensor. 
 The time complexity of the construction  is linear in $|V|$, the number of nodes in the constructed TDD.
 

\begin{algorithm}
\caption{$\TDD\_generate(\phi)$}
\begin{algorithmic}[1]
\Require{A tensor $\phi$ over a linearly ordered index set $I$.}
\Ensure{The reduced TDD of $\phi$.}

\If{$\phi \equiv c$ is a constant}
\State \Return the trivial TDD with weight $c$
\EndIf
\State  $x\leftarrow$ the smallest index of $\phi$
\State $tdd\leftarrow$ an empty TDD
\State $tdd.root \leftarrow$ a new node $v$  with index $x$
\State $v.low\leftarrow \TDD\_generate(\phi|_{x=0})$
\State $v.high\leftarrow \TDD\_generate(\phi|_{x=1})$
\State $tdd.weight\leftarrow 1$
\State \Return $reduce(tdd)$
\end{algorithmic}
\label{TDD_generate}
\end{algorithm}

We emphasise that, if an index is repeated in the tensor, for example $\phi_{xxy}$, then the two successors of the node representing $\phi_{xxy}$  will be $\phi_{00y}$ and $\phi_{11y}$. In other words, we construct the TDD as if it is the tensor $\phi_{xy}$. When tensor operations are concerned, however, both $x$ indices will be involved.

\begin{figure}
    \centering
    \subfigure[]{
     \centering
    \includegraphics[width=0.27\textwidth]{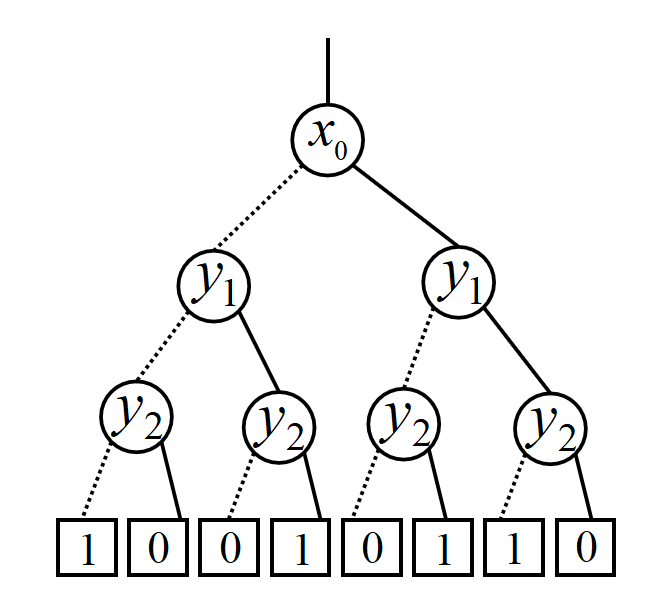}
    }
    \subfigure[]{
    \centering
    \includegraphics[width=0.275\textwidth]{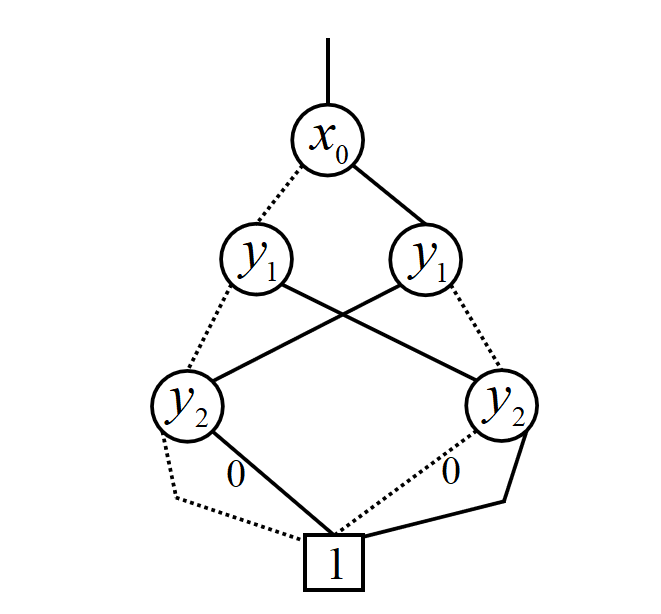}
    }
    \caption{Two TDDs of the \cnot\ gate with indices $x_0,y_1,y_2$: (a) the general form before normalisation and reduction; (b) the reduced TDD, where nodes are normalised and those represent the same tensor are merged. 
    }
    \label{exp-for-TDD2}
\end{figure}

\begin{example}
Consider the \cnot\ gate shown in Fig.~\ref{exp-for-hyper-edge} (right), which is represented by a tensor $\phi_{x_0x_0y_1y_2}$. The  reduced TDD of  $\phi_{x_0x_0y_1y_2}$ is shown in Fig.~\ref{exp-for-TDD2}(b), where the index $x_0$ only appears once with the two successors representing the tensor $\phi_{00y_1y_2}$ and $\phi_{11y_1y_2}$. 
\end{example}


\subsection{Addition}
Let $\F$ and $\g$ be two reduced TDDs over index set $I$. The summation of $\F$ and $\g$, denoted $\F+\g$, is a reduced TDD with the corresponding tensor $\bphi(\F)+\bphi(\g)$. For any $x\in I$ with $x\preceq index(r_\F)$ and $x\preceq index(r_\g)$, by the TDD version of the Boole-Shannon expansion (cf. Eq.~\ref{eq:tdd-shannon}), we have  
\begin{align*}
\bphi(\F)+ \bphi(\g) = &\  \overline{x} \cdot (\bphi(\F_{x=0})+ \bphi(\g_{x=0})) + x \cdot (\bphi(\F_{x=1}) + \bphi(\g_{x=1})).
\end{align*}
Recall here $\F_{x=c}$ (resp. $\g_{x=c}$) is the sub-TDD as defined in Definition~\ref{def:subtdd} for $c\in \{0,1\}$.

Motivated by this observation,  Algorithm~\ref{TDD_Add} implements the $Add$ operation for TDDs, in a node-wise manner. \xblue{For two TDDs $\F$ and $\g$ as above, the size of $\F+\g$ is related to the number of paths in $\F$ and $\g$.
}


\begin{algorithm}
\caption{$Add(\F, \g)$}
\begin{algorithmic}[1]
\Require{
Two reduced TDDs $\F$ and $\g$}.
\Ensure {The reduced TDD of $\bphi(\F)+\bphi(\g)$}.
\If{$r_\F= r_\g$}
\State $tdd\leftarrow \F$
\State $tdd.weight\leftarrow w_\F\ + \ w_\g$
\State \Return $tdd$
\EndIf
\State $x\leftarrow$ the smaller index of $r_\F$ and $r_\g$
\State $tdd\leftarrow$ an empty TDD
\State $tdd.root \leftarrow$ a new node $v$  with index $x$
\State $v.low\leftarrow Add(\F_{x=0},\g_{x=0})$
\State $v.high\leftarrow Add(\F_{x=1},\g_{x=1})$
\State $tdd.weight\leftarrow 1$
\State \Return $reduce(tdd)$

\end{algorithmic}
\label{TDD_Add}
\end{algorithm}

\subsection{Contraction}\label{sec:contraction}
Contraction is the  most fundamental operation in a tensor network. Many design automation tasks of quantum circuits are based on contraction. In this subsection, we consider how to efficiently implement the contraction operation via TDD. 

Let  $\F$ and $\g$ be two reduced TDDs over $I$, and $var$ a subset of $I$ denoting the variables to be contracted. Write  $\cont$ for both tensor and TDD contractions. For any $x\in I$ with $x\preceq index(r_\F)$ and $x\preceq index(r_\g)$, we have by definition Eq.~\ref{eq:contdef} that if $x\in var$, then $\cont\left(\bphi(\F),\bphi(\g), var\right)$ equals
	$$\sum_{c=0}^{1} \cont(\bphi(\F_{x=c}),\bphi(\g_{x=c}), var\backslash\{x\});$$
otherwise, it equals
\begin{align*}
& \overline{x} \cdot \cont(\bphi(\F_{x=0}),\bphi(\g_{x=0}), var)
	 + x \cdot \cont(\bphi(\F_{x=1}),\bphi(\g_{x=1}), var).
\end{align*}
%
%
%
%
%
%

Algorithm~\ref{TDD_Contraction} gives the detailed procedure for TDD contraction. The time complexity \xblue{and the size of the contracted TDD are related to the number of paths in $\F$ and $\g$.
}

To conclude this section, we would like to point out that the \emph{tensor product} of two TDDs $\F$ and $\g$ with disjoint essential indices can be regarded as a special case of contraction.
In particular, we have
$$\bphi(\F\otimes \g) = \cont(\bphi(\F),\bphi(\g), \emptyset),$$
and the time complexity of using Algorithm~\ref{TDD_Contraction} to compute $\F\otimes \g$ becomes $|\F|\cdot |\g|$, where $|\F|$ and $|\g|$ denote the numbers of nodes in $\F$ and $\g$, respectively..

A special case which arises often in applications is when, say, every index in $\F$ precedes any index in $\g$ under the order $\prec$.
%
For this case, to compute the tensor product of $\F$ and $\g$, all we need to do is to replace the terminal node of $\F$ with the root node of $\g$, multiply the weight of the resulting TDD with the weight of $\g$, and perform normalisation and reduction if necessary.
Since we do not need to touch $\g$, the time complexity is simply $\mathcal{O}(|\F|)$.
%


\begin{algorithm}
\caption{$\cont(\F, \g, var)$}
\begin{algorithmic}[1]
\Require {Two reduced TDDs $\F$ and $\g$, and the set $var$ of variables to be contracted.}
\Ensure {The reduced TDD obtained by contracting $\F$ and $\g$ over $var$.}
\If{both $\F$ and $\g$ are trivial}
\State $tdd\leftarrow \F$
\State $tdd.weight \leftarrow w_\F \cdot w_\g \cdot 2^{len(var)}$
\State \Return $tdd$
\EndIf
\State $x\leftarrow$ the smaller index of $r_\F$ and $r_\g$
\State $L\leftarrow \cont(\F_{x=0},\g_{x=0},var\backslash\{x\})$
\State $R\leftarrow \cont(\F_{x=1},\g_{x=1},var\backslash\{x\})$
\If{$x\in var$}
\State \Return $Add(L,R)$
\Else
\State $tdd\leftarrow$ an empty TDD
\State $tdd.root \leftarrow$ a new node $v$ with index $x$
\State $v.low\leftarrow L$
\State $v.high\leftarrow R$
\State $tdd.weight\leftarrow 1$
\State \Return $reduce(tdd)$
\EndIf
\end{algorithmic}
\label{TDD_Contraction}
\end{algorithm}

\section{Two partition schemes} \label{sec:circuitpartition}

The TDD representation of a quantum circuit can be calculated flexibly. In particular, there is no need to expand a quantum gate to an $n$-qubit form (by tensoring an identity matrix). In general, the TDD representation of a quantum circuit can be obtained by contracting the TDDs of individual gates in the circuit in any order.   

In this paper, we assume that all gates in our circuits are either single-qubit gates or \cnot\ gates. For simplicity of presentation, we use the original qubit order (or its inverse). Following this order, we scan the circuit qubit by qubit, and then rank the indices following the circuit order. That is, given two indices $x$ and $x'$ appearing in the circuit, suppose $q_i$ and $q_j$ are the qubits that $x$ and $x'$ are on. Then we set $x\prec x'$ if either $i<j$, or $i=j$ and $x$ is to the left of $x'$ on the qubit wire $q_i$. For example, the selected order 
for the circuit shown in Fig.~\ref{exp-for-par1} is
\begin{equation}\label{eq:prec}
\prec\ :=\ (x_0,x_{0,1},x_{0,2}, y_0, x_{1}, x_{1,1}, x_{1,2}, y_1, x_2, ..., y_2, x_3, ...,y_3).
\end{equation}
\xblue{Note that hyper-edges  (cf. Fig.~\ref{exp-for-hyper-edge}(b)) are used here. For example, $x_{1,1}$ passes a \cx\ gate, a diagonal $T$ gate and another \cx\ gate.}
\blue{It is worth stressing that, once such an index order (like the one in Eq.~\ref{eq:prec}) is selected, it will be used for all TDDs generated in the calculation process.}

Our approach of computing the TDD of a quantum circuit includes two steps. First, we partition the circuit into several parts; second, we calculate the TDD of each part separately and then combine them together through contraction.

\begin{figure}
\centering
\includegraphics[width=0.7\textwidth]{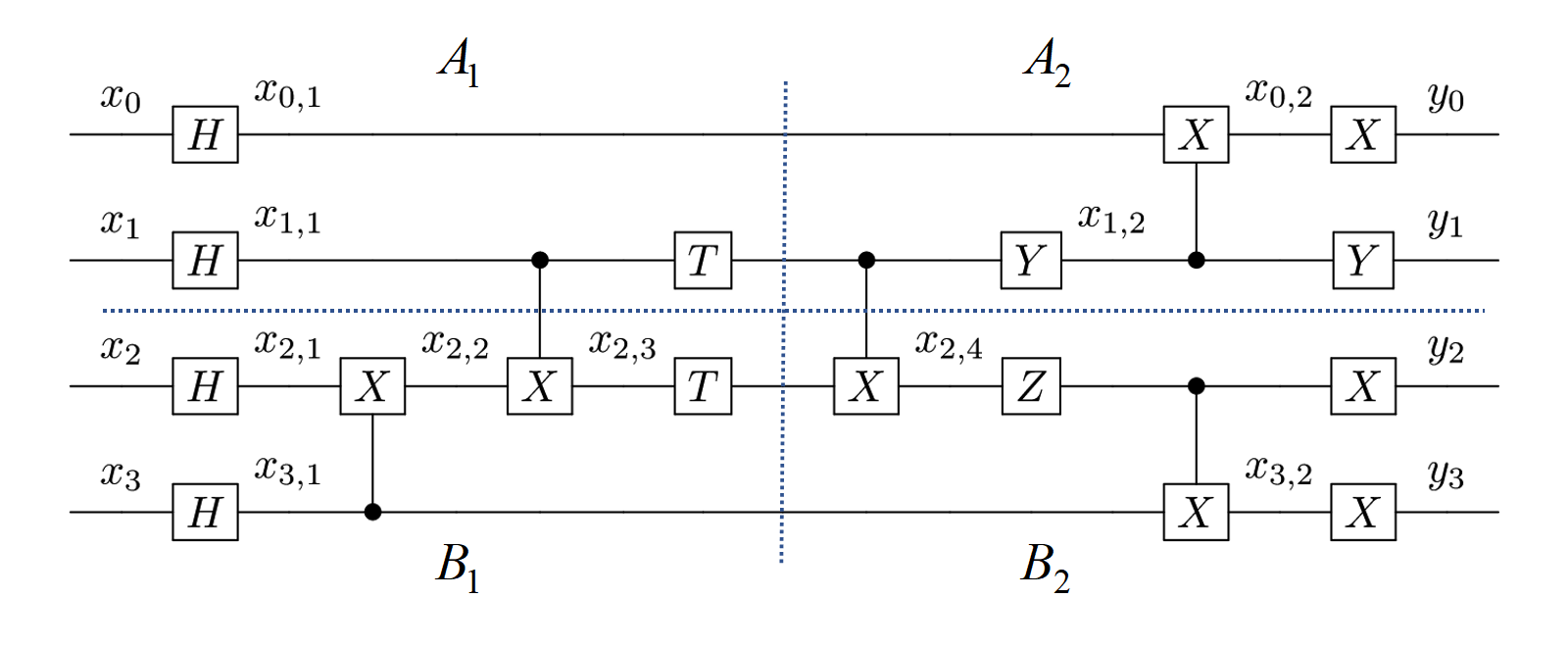}
\caption{Partition Scheme I, where only one \cnot\ cut is allowed each time.}
\label{exp-for-par1}
\end{figure}

While finding the optimal partition scheme is attractive, it is also a very challenging task.  We observe that some simple strategies are already able to reduce the resource consumption significantly during the computation process. In the following, we introduce two straightforward partition schemes.



The first partition scheme divides the circuit \emph{horizontally} into two parts from the middle (so that the upper and lower parts have roughly the same number of qubits) and then cuts it \emph{vertically} 
such that in each part no more than $k$ (a predefined parameter) \cnot\ gates are separated by the horizontal cut, where $k$ is chosen to ensure that the rank of each block of the final circuit is smaller than $2n$, the rank of the tensor of the original circuit. In our experiments, we set $k=\lfloor{n/2}\rfloor$.

 


\begin{example}
Consider the circuit shown in Fig.~\ref{exp-for-par1} 
and set $k=1$, i.e., we allow only one \cnot\ cut at a time. The circuit is divided into four parts as shown by the dotted lines. In the contraction process, we first calculate the TDDs of the four parts separately. Then, contracting the left (right, resp.) two TDDs gives the TDD of the left  (right, resp.) half of the circuit. Finally, we contract these two TDDs and obtain the TDD of the whole circuit. If we set $k=2$, then no vertical cut is required and the circuit is partitioned into two parts: the top half and the bottom half. The same TDD can be obtained by contracting the TDDs of the top and the bottom halves. 
\end{example}

\begin{figure}
\centering
\includegraphics[width=0.7\textwidth]{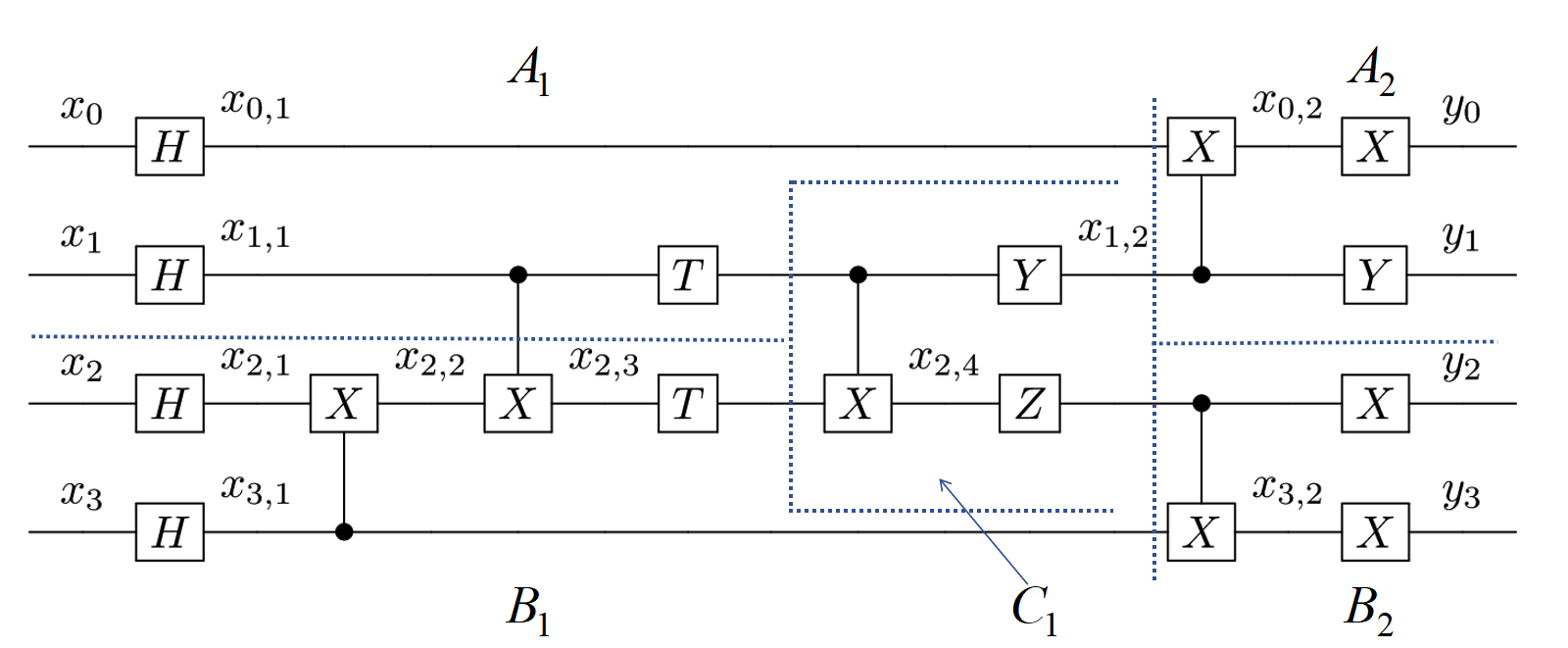}
\caption{Partition Scheme II, where only one \cnot\ cut is allowed each time and part $C$ \rmagenta{involves} up to 2 qubits.}
\label{exp-for-par2}
\end{figure}

Note that the more \cnot\ gates separated by the horizontal cut, the more vertical partitions we need to introduce, and the more tensors with large (near $2n$) rank we need to contract.  Our second partition scheme intends to alleviate this issue by enveloping those separated \cnot\ gates if they are closely located in the circuit. First, we horizontally divide the circuit from the middle as in the first scheme. Whenever $k_1$ (a predefined parameter) \cnot\ gates have been separated by the horizontal cut, we introduce a third small block that wraps a small part of the circuit such that it contains several \cnot\ gates that are separated by the horizontal cut and has no interactions (i.e., shares no \cnot)  with qubits not occupied by that gates in this block. Whenever the third block occupies $k_2$ (a predefined parameter) qubits, we introduce a new vertical cut. The second scheme is a generalisation of the scheme used in \cite{li2019quantum} for classical simulation, where no vertical partition is introduced. In our experiments, we set $k_1$ and $k_2$ as $\lfloor{n/2}\rfloor$  and $\lfloor{n/2}\rfloor+1$, respectively.

\begin{example}
Consider the circuit given in Fig.~\ref{exp-for-par1} again. Suppose we allow one \cx\ cut every time, and limit the number of qubits in part $C$ to two. Then the circuit can be partitioned into five parts as illustrated in Fig.~\ref{exp-for-par2}. 
\xblue{We note that, in this scheme, a third block $C_1$ is introduced only when we have separated $k_1$ (here it is 1) \cx\ gates. A better partition can be obtained by incorporating the first \cx\ gate between $q_1$ and $q_2$, as well as the following two $T$ gates, into $C_1$.}
We then compute and contract the TDDs in the order of $A$, $B$, $C$ for every block split by the vertical lines. The TDD of the whole circuit is then obtained by contracting the TDDs of these blocks in sequence.
\end{example}

Now we make a simple comparison of the above contraction methods. Suppose we compute the TDD (or QMDD) representation of the circuit in Fig.~\ref{exp-for-par1} in the original circuit order. We need in essence to calculate eight $(8, 2, 1)$-contractions, five $(8, 4, 2)$-contractions, one $(6,2,0)$-contraction, and two  contractions between tensors with rank $\leq 4$, where an $(m,n,r)$-contraction is a contraction between a rank $m$ tensor and a rank $n$ tensor over $r$ common indices. In comparison, Partition Scheme I requires one $(8,8,4)$-contraction, two $(5,5,1)$-contractions, five $(5,2,1)$-contractions, and nine contractions between tensors with rank $\leq 4$; while Partition Scheme II requires one $(8,8,4)$-contraction, one $(8,4,2)$-contraction, one $(5,5,1)$-contraction, and 14 contractions between tensors with rank $\leq 4$. As the time and space consumption both grow exponentially with the ranks of the tensors \cite{gray2020hyper}, this illustrates the efficiency of the two partition schemes.


\section{Implementation and Evaluation} \label{sec:experiments}


To demonstrate the effectiveness of TDD as an efficient data structure for the representation and manipulation of quantum functionalities, we developed the TDD package in Python3, implemented the two partition schemes, and empirically compared them with three state-of-the-art approaches in the literature. 

\subsection{Benchmarks}
Most benchmarks we used were published by IBM as part of the 2018 QISKit Developer Challenge\footnote{https://www.ibm.com/blogs/research/2018/08/winners-qiskit-developer-challenge/}, which\footnote{Available from \url{http://iic.jku.at/eda/research/ibm_qx_mapping/}} 
have been wildly used in evaluating qubit mapping algorithms (see, e.g., \cite{AstarZulehner}).
%
%
%
To compare the scalability of different methods, we also tested three commonly used quantum algorithms, including  Bernstein-Vazirani (bv) \cite{bernstein1997quantum}, Quantum Fourier Transform (qft) \cite{nielsen2002quantum}, as well as Quantum Volume (qv) \cite{moll2018quantum}. 
The numbers of qubits and gates in these benchmarks range from 2 to 100 and 5 to 10,223, respectively.


\subsection{TDD Implementation} 

We implemented TDD using Python3. In our calculation process, calculated results are stored in a computed table as in the implementation of ROBDD \cite{brace1990efficient}. In order to improve the reusability of the calculated results in the computed table, we map all indices of a TDD to $\{0,1,\cdots,m-1\}$, where $m$ is the number of different indices of the tensor associated to the TDD, such that TDDs differ only by a renaming of indices will be treated as the same. Our source code is available at Github.\footnote{https://github.com/VeriQC/TDD} 



In our experiments, for the first partition scheme, we set the parameter $k$ as $\lfloor n/2 \rfloor$, where $n$ is the number of qubits in the input circuit. Similarly, for the second partition scheme, we set the two parameters $k_1$ and $k_2$ as $\lfloor n/2 \rfloor$ and $\lfloor n/2 \rfloor+1$, respectively. All experiments were executed on a laptop with Intel i7-1065G7 CPU and 8 GB RAM.

\subsection{Empirical results}
We compared our results with three state-of-the-art approaches for computing quantum functionalities --- Qiskit (https://qiskit.org), the Google TensorNetwork package \cite{roberts2019tensornetwork}, and QMDD  \cite{niemann2015qmdds}. The first two approaches are matrix-based, while QMDD is decision diagram-based. For Qiskit, we call the unitary\_simulator for calculating the unitary matrix of every circuit, and for TensorNetwork, a tensor network is constructed for every circuit and the auto contractor will be used for completing the task. For QMDD, we compute the functionality of an input circuit in a way similar to TDD with no partition, i.e., we construct the QMDD of each quantum gate and then multiply them in the circuit order. \xblue{The QMDD package we used is the version obtained from https://github.com/iic-jku/dd\_package.}


\begin{table}[]
\centering
\caption{Data summary for benchmark circuits taken from \cite{AstarZulehner}, where Time (-MO) denotes the total time of all circuits on which TN (the Google TensorNetwork) is not memory out.}

\setlength{\tabcolsep}{1mm}{
\scalebox{0.9}{
\begin{tabular}{c|c|ccc|c}
                   & \multirow{2}{*}{QMDD} & \multicolumn{3}{c}{TDD} & \multirow{2}{*}{TN} \\ 
                   &                       & No part. & Part. I & Part. II& \\ \hline
Time (-MO) & 2.30 & 259.38  & 117.51 & 88.45 &115.15 \\
 \hline
 Time   & 2.39  & 263.67 & 119.56 & 90.53 &- \\
node num. (final$^*$) & 6413 & 13888   & 13888 & 13888  &-        \\
node num. (max$^*$)   & 15758 &  36194  &    17769 &  17325  &-      \\
ratio (max/final$^*$) & 2.46        & 2.61      & 1.28          & 1.25 &-
\\
\end{tabular}}
}
\label{experiment-results}
\end{table}

We summarise our experimental results on the benchmark circuits from \cite{AstarZulehner} in Table~\ref{experiment-results}. More and detailed results can be found in Table~\ref{full-experiment-results}. The performance of Qiskit is very similar but inferior to TensorNetwork. Except the 15 qubit circuit `rd84\_142', 
Qiskit can finish in 116.7s all circuits which TensorNetwork does not run out of memory. For `rd84\_142', Qiskit runs out of memory but TensorNetwork finishes in 49.3s.
In the following, we omit the results of Qiskit from Table~\ref{full-experiment-results}.





\vspace*{3mm}

\xblue{
\noindent{\textbf{Notes of Table~\ref{full-experiment-results}:}}
\begin{itemize}
\item[1.] TN represents the Google Tensor network package; MO and TO represent, respectively, out of memory and time out of 3600 seconds; sum(-MO)  (sum(-TO), resp.) represents the sum with MO (TO, resp.) circuits above the line it starts
being excluded.
\item[2.] For TDD and QMDD, we list the time (seconds), max number of nodes (\#node max) and final number of nodes (\#node final) in the construction process. For TDD with the two partition schemes, we remove the final number of nodes, as they are all identical to that of TDD with no partition. 
\end{itemize}
}

\renewcommand{\arraystretch}{0.85}
{

\begin{table}[]
\centering
\caption{Experiment Data}

\setlength{\tabcolsep}{0.3mm}{
\scalebox{0.75}{
\begin{tabular}{lcccccccccccccccccc}
\hline
\multicolumn{3}{c}{Benchmarks}                                                                                             & \multicolumn{1}{c}{} & \multicolumn{3}{c}{QMDD}                                                                                                                                                               &  & \multicolumn{3}{c}{TDD No Part.}                                                                                                                                                                & \multicolumn{1}{c}{} & \multicolumn{2}{c}{TDD Part. I}                                                                               & \multicolumn{1}{c}{} & \multicolumn{2}{c}{TDD Part. II}                                                                               & \multicolumn{1}{c}{} & \multicolumn{1}{c}{TN}   \\ \cline{1-3} \cline{5-7} \cline{9-11} \cline{13-14} \cline{16-17} \cline{19-19} 
circ. name           & \#qubit & \#gate &                      & \multicolumn{1}{c}{time} & \multicolumn{1}{c}{\begin{tabular}[c]{@{}c@{}}\#node\\  max\end{tabular}} & \multicolumn{1}{c}{\begin{tabular}[c]{@{}c@{}}\#node\\  final\end{tabular}} &  & \multicolumn{1}{c}{time} & \multicolumn{1}{c}{\begin{tabular}[c]{@{}c@{}}\#node\\ max\end{tabular}} & \multicolumn{1}{c}{\begin{tabular}[c]{@{}c@{}}\#node\\ final\end{tabular}} & \multicolumn{1}{c}{} & \multicolumn{1}{c}{time} & \multicolumn{1}{c}{\begin{tabular}[c]{@{}c@{}}\#node\\ max\end{tabular}} & \multicolumn{1}{c}{} & \multicolumn{1}{c}{time} & \multicolumn{1}{c}{\begin{tabular}[c]{@{}c@{}}\#node\\ max\end{tabular}} & \multicolumn{1}{c}{} & \multicolumn{1}{c}{time} \\ \hline
graycode6\_47  & 6                                                    & 5                                                  &                      & 0.03                     & 12                                                                          & 12                                                                            &  & 0.01                     & 22                                                                          & 22                                                                            &                      & 0.01                     & 22                                                                          &                      & 0.01                     & 22                                                                          &                      & 0.01                     \\
ex-1\_166      & 3                                                    & 19                                                 &                      & 0.03                     & 10                                                                          & 9                                                                             &  & 0.02                     & 22                                                                          & 17                                                                            &                      & 0.03                     & 22                                                                          &                      & 0.04                     & 22                                                                          &                      & 0.01                     \\
4mod5-v0\_20   & 5                                                    & 20                                                 &                      & 0.02                     & 28                                                                          & 22                                                                            &  & 0.03                     & 44                                                                          & 36                                                                            &                      & 0.03                     & 38                                                                          &                      & 0.03                     & 38                                                                          &                      & 0.02                     \\
rd32-v0\_66    & 4                                                    & 34                                                 &                      & 0.03                     & 14                                                                          & 9                                                                             &  & 0.04                     & 25                                                                          & 20                                                                            &                      & 0.03                     & 23                                                                          &                      & 0.04                     & 28                                                                          &                      & 0.03                     \\
decod24-v0\_38 & 4                                                    & 51                                                 &                      & 0.03                     & 18                                                                          & 16                                                                            &  & 0.11                     & 38                                                                          & 35                                                                            &                      & 0.06                     & 38                                                                          &                      & 0.08                     & 38                                                                          &                      & 0.04                     \\
4gt13\_92      & 5                                                    & 66                                                 &                      & 0.03                     & 26                                                                          & 26                                                                            &  & 0.12                     & 58                                                                          & 58                                                                            &                      & 0.12                     & 58                                                                          &                      & 0.11                     & 58                                                                          &                      & 0.04                     \\
4mod5-bdd\_287 & 7                                                    & 70                                                 &                      & 0.04                     & 111                                                                         & 74                                                                            &  & 0.26                     & 155                                                                         & 109                                                                           &                      & 0.24                     & 128                                                                         &                      & 0.17                     & 128                                                                         &                      & 0.04                     \\
alu-v0\_26     & 5                                                    & 84                                                 &                      & 0.04                     & 46                                                                          & 23                                                                            &  & 0.25                     & 68                                                                          & 49                                                                            &                      & 0.14                     & 68                                                                          &                      & 0.15                     & 68                                                                          &                      & 0.05                     \\
4gt5\_76       & 5                                                    & 91                                                 &                      & 0.03                     & 50                                                                          & 20                                                                            &  & 0.30                     & 73                                                                          & 34                                                                            &                      & 0.15                     & 73                                                                          &                      & 0.15                     & 73                                                                          &                      & 0.06                     \\
4gt5\_77       & 5                                                    & 131                                                &                      & 0.03                     & 46                                                                          & 27                                                                            &  & 0.61                     & 91                                                                          & 57                                                                            &                      & 0.29                     & 91                                                                          &                      & 0.24                     & 77                                                                          &                      & 0.07                     \\
decod24-v3\_45 & 5                                                    & 150                                                &                      & 0.04                     & 46                                                                          & 17                                                                            &  & 0.41                     & 74                                                                          & 35                                                                            &                      & 0.25                     & 64                                                                          &                      & 0.23                     & 82                                                                          &                      & 0.06                     \\
cnt3-5\_179    & 16                                                   & 175                                                &                      & 0.04                     & 55                                                                          & 48                                                                            &  & 1.00                     & 116                                                                         & 104                                                                           &                      & 0.51                     & 148                                                                         &                      & 0.45                     & 148                                                                         &                      & MO                       \\
0410184\_169   & 14                                                   & 211                                                &                      & 0.03                     & 63                                                                          & 39                                                                            &  & 0.87                     & 117                                                                         & 81                                                                            &                      & 0.47                     & 125                                                                         &                      & 0.39                     & 102                                                                         &                      & 3.59                     \\
sys6-v0\_111   & 10                                                   & 215                                                &                      & 0.05                     & 473                                                                         & 247                                                                           &  & 3.50                     & 877                                                                         & 562                                                                           &                      & 1.58                     & 685                                                                         &                      & 1.00                     & 562                                                                         &                      & 0.14                     \\
4gt4-v0\_72    & 6                                                    & 258                                                &                      & 0.03                     & 98                                                                          & 31                                                                            &  & 1.01                     & 136                                                                         & 68                                                                            &                      & 0.50                     & 123                                                                         &                      & 0.51                     & 119                                                                         &                      & 0.14                     \\
sym6\_316      & 14                                                   & 270                                                &                      & 0.09                     & 5608                                                                        & 1520                                                                          &  & 11.55                    & 10836                                                                       & 3028                                                                          &                      & 2.70                     & 3028                                                                        &                      & 2.65                     & 3028                                                                        &                      & 2.35                     \\
sym9\_146      & 12                                                   & 328                                                &                      & 0.06                     & 523                                                                         & 229                                                                           &  & 6.66                     & 1496                                                                        & 515                                                                           &                      & 2.60                     & 515                                                                         &                      & 2.62                     & 515                                                                         &                      & 0.31                     \\
mod8-10\_178   & 6                                                    & 342                                                &                      & 0.03                     & 72                                                                          & 18                                                                            &  & 0.98                     & 144                                                                         & 40                                                                            &                      & 0.73                     & 120                                                                         &                      & 0.63                     & 140                                                                         &                      & 0.17                     \\
rd84\_142      & 15                                                   & 343                                                &                      & 0.26                     & 6922                                                                        & 3588                                                                          &  & 56.63                    & 18006                                                                       & 8017                                                                          &                      & 15.69                    & 9175                                                                        &                      & 8.98                     & 9095                                                                        &                      & 49.33                    \\
alu-v2\_31     & 5                                                    & 451                                                &                      & 0.04                     & 44                                                                          & 18                                                                            &  & 1.03                     & 93                                                                          & 44                                                                            &                      & 0.69                     & 87                                                                          &                      & 0.77                     & 91                                                                          &                      & 0.17                     \\
cnt3-5\_180    & 16                                                   & 485                                                &                      & 0.05                     & 164                                                                         & 48                                                                            &  & 3.29                     & 349                                                                         & 104                                                                           &                      & 1.54                     & 319                                                                         &                      & 1.60                     & 311                                                                         &                      & MO                       \\
rd53\_133      & 7                                                    & 580                                                &                      & 0.04                     & 78                                                                          & 26                                                                            &  & 1.94                     & 136                                                                         & 59                                                                            &                      & 0.97                     & 123                                                                         &                      & 0.98                     & 123                                                                         &                      & 0.35                     \\
majority\_239  & 7                                                    & 612                                                &                      & 0.06                     & 87                                                                          & 16                                                                            &  & 2.08                     & 231                                                                         & 39                                                                            &                      & 1.31                     & 231                                                                         &                      & 1.26                     & 177                                                                         &                      & 0.37                     \\
sf\_276        & 6                                                    & 778                                                &                      & 0.06                     & 59                                                                          & 19                                                                            &  & 2.17                     & 100                                                                         & 42                                                                            &                      & 1.32                     & 90                                                                          &                      & 1.12                     & 100                                                                         &                      & 0.45                     \\
con1\_216      & 9                                                    & 954                                                &                      & 0.07                     & 214                                                                         & 37                                                                            &  & 4.85                     & 373                                                                         & 91                                                                            &                      & 3.10                     & 329                                                                         &                      & 2.25                     & 242                                                                         &                      & 1.00                     \\
cm42a\_207     & 14                                                   & 1776                                               &                      & 0.09                     & 211                                                                         & 81                                                                            &  & 7.69                     & 639                                                                         & 216                                                                           &                      & 7.43                     & 460                                                                         &                      & 7.91                     & 460                                                                         &                      & 17.23                    \\
hwb6\_56       & 7                                                    & 6723                                               &                      & 0.38                    & 242                                                                         & 88                                                                            &  & 53.91                    & 824                                                                         & 178                                                                           &                      & 29.99                    & 612                                                                         &                      & 27.73                    & 731                                                                         &                      & 10.46                    \\
sqn\_258       & 10                                                   & 10223                                              &                      & 0.67                    & 438                                                                         & 105                                                                           &  & 102.34                   & 1051                                                                        & 228                                                                           &                      & 47.09                    & 974                                                                         &                      & 28.45                    & 747                                                                         &                      & 28.66                    \\ \hline

sum(-MO) & -                                                   & -                                              &                      & 2.30                    & -                                                                         & -                                                                           &  & 259.38                   & -                                                                        & -                                                                           &                      & 117.51                    & -                                                                         &                      & 88.45                    & -                                                                         &                      & 115.15   \\
sum & -                                                   & -                                              &                      & 2.39                    & 15758                                                                         & 6413                                                                           &  & 263.67                   & 36194                                                                        & 13888                                                                           &                      & 119.56                   & 17769                                                                         &                      & 90.53                    & 17325                                                                         &                      & -    \\ 
\hline\hline
bv\_10         & 10                                                   & 29                                                 &                      & 0.04                     & 20                                                                          & 20                                                                            &  & 0.09                     & 56                                                                          & 56                                                                            &                      & 0.06                     & 56                                                                          &                      & 0.06                     & 56                                                                          &                      & 0.02                     \\
bv\_20         & 20                                                   & 59                                                 &                      & 0.04                     & 40                                                                          & 40                                                                            &  & 0.29                     & 116                                                                         & 116                                                                           &                      & 0.21                     & 116                                                                         &                      & 0.21                     & 116                                                                         &                      & MO                       \\
bv\_30         & 30                                                   & 89                                                 &                      & 0.04                     & 60                                                                          & 60                                                                            &  & 0.70                     & 176                                                                         & 176                                                                           &                      & 0.66                     & 176                                                                         &                      & 0.49                     & 176                                                                         &                      & MO                       \\
bv\_40         & 40                                                   & 119                                                &                      & 0.04                     & 80                                                                          & 80                                                                            &  & 1.22                     & 236                                                                         & 236                                                                           &                      & 0.84                     & 236                                                                         &                      & 0.81                     & 236                                                                         &                      & MO                       \\
bv\_50         & 50                                                   & 149                                                &                      & 0.05                    & 100                                                                         & 100                                                                           &  & 1.74                     & 296                                                                         & 296                                                                           &                      & 1.34                     & 296                                                                         &                      & 1.26                     & 296                                                                         &                      & MO                       \\
bv\_60         & 60                                                   & 179                                                &                      & 0.06                     & 120                                                                         & 120                                                                           &  & 2.82                     & 356                                                                         & 356                                                                           &                      & 1.88                     & 356                                                                         &                      & 1.77                     & 356                                                                         &                      & MO                       \\
bv\_70         & 70                                                   & 209                                                &                      & 0.06                     & 140                                                                         & 140                                                                           &  & 4.45                     & 416                                                                         & 416                                                                           &                      & 3.18                     & 416                                                                         &                      & 2.36                     & 416                                                                         &                      & MO                       \\
bv\_80         & 80                                                   & 239                                                &                      & 0.07                     & 160                                                                         & 160                                                                           &  & 4.59                     & 476                                                                         & 476                                                                           &                      & 3.76                     & 476                                                                         &                      & 3.58                     & 476                                                                         &                      & MO                       \\
bv\_90         & 90                                                   & 269                                                &                      & 0.08                     & 180                                                                         & 180                                                                           &  & 5.91                     & 536                                                                         & 536                                                                           &                      & 4.14                     & 536                                                                         &                      & 5.63                     & 536                                                                         &                      & MO                       \\
bv\_100        & 100                                                  & 299                                                &                      & 0.09                     & 200                                                                         & 200                                                                           &  & 7.38                     & 596                                                                         & 596                                                                           &                      & 4.86                     & 596                                                                         &                      & 6.96                     & 596                                                                         &                      & MO                       \\\hline\hline
qv\_n2\_d5     & 2                                                    & 50                                                 &                      & 0.04                     & 6                                                                           & 6                                                                             &  & 0.17                     & 16                                                                          & 16                                                                            &                      & 0.12                     & 16                                                                          &                      & 0.14                     & 16                                                                          &                      & 0.05                     \\
qv\_n3\_d5     & 3                                                    & 50                                                 &                      & 0.04                     & 22                                                                          & 22                                                                            &  & 0.18                     & 64                                                                          & 64                                                                            &                      & 0.15                     & 64                                                                          &                      & 0.16                     & 64                                                                          &                      & 0.02                     \\
qv\_n4\_d5     & 4                                                    & 100                                                &                      & 0.06                     & 86                                                                          & 86                                                                            &  & 1.27                     & 256                                                                         & 256                                                                           &                      & 0.89                     & 256                                                                         &                      & 1.03                     & 256                                                                         &                      & 0.03                     \\
qv\_n5\_d5     & 5                                                    & 100                                                &                      & 0.09                     & 342                                                                         & 342                                                                           &  & 3.74                     & 1024                                                                        & 1024                                                                          &                      & 2.26                     & 1024                                                                        &                      & 2.60                     & 1024                                                                        &                      & 0.13                     \\
qv\_n6\_d5     & 6                                                    & 150                                                &                      & 0.21                     & 1366                                                                        & 1366                                                                          &  & 17.71                    & 4096                                                                        & 4096                                                                          &                      & 9.18                     & 4096                                                                        &                      & 10.59                    & 4096                                                                        &                      & 0.05                     \\
qv\_n7\_d5     & 7                                                    & 150                                                &                      & 3.51                    & 5462                                                                        & 5462                                                                          &  & 81.97                    & 16384                                                                       & 16384                                                                         &                      & 35.48                    & 16384                                                                       &                      & 37.56                    & 16384                                                                       &                      & 0.07                     \\
qv\_n8\_d5     & 8                                                    & 200                                                &                      & 77.39	&21846	&21846
                                                                          &  & 306.52&	65536	&65536
                                                                         &                      & 412.99	&65536
                                                                  &                      &456.98&	65536
                                                                   &                      & 0.07
\\
qv\_n9\_d5     & 9                                                    & 200                                                &                      & 966.62	&87382	&87382
                                                                          &  & TO&		&
                                                                         &                      & 2004.36 
	&262144

                                                                  &                      &2195.29 
&262144

                                                                   &                      & 0.08
\\
\hline\hline
qft\_5         & 5                                                    & 15                                                 &                      & 0.02                     & 32                                                                          & 32                                                                            &  & 0.02                     & 63                                                                          & 63                                                                            &                      & 0.01                     & 63                                                                          &                      & 0.02                     & 63                                                                          &                      & 0.09                     \\
qft\_6         & 6                                                    & 21                                                 &                      & 0.03                     & 64                                                                          & 64                                                                            &  & 0.03                     & 127                                                                         & 127                                                                           &                      & 0.02                     & 127                                                                         &                      & 0.02                     & 127                                                                         &                      & 0.03                     \\
qft\_7         & 7                                                    & 28                                                 &                      & 0.03                     & 128                                                                         & 128                                                                           &  & 0.07                     & 255                                                                         & 255                                                                           &                      & 0.03                     & 255                                                                         &                      & 0.04                     & 255                                                                         &                      & 0.04                     \\
qft\_8         & 8                                                    & 36                                                 &                      & 0.03                     & 256                                                                         & 256                                                                           &  & 0.10                     & 511                                                                         & 511                                                                           &                      & 0.05                     & 511                                                                         &                      & 0.05                     & 511                                                                         &                      & 0.05                     \\
qft\_9         & 9                                                    & 45                                                 &                      & 0.03                     & 512                                                                         & 512                                                                           &  & 0.15                     & 1023                                                                        & 1023                                                                          &                      & 0.08                     & 1023                                                                        &                      & 0.08                     & 1023                                                                        &                      & 0.06                     \\
qft\_10        & 10                                                   & 55                                                 &                      & 0.04                     & 1024                                                                        & 1024                                                                          &  & 0.37                     & 2047                                                                        & 2047                                                                          &                      & 0.14                     & 2047                                                                        &                      & 0.17                     & 2047                                                                        &                      & 0.12                     \\
qft\_11        & 11                                                   & 66                                                 &                      & 0.05                     & 2048                                                                        & 2048                                                                          &  & 0.71                     & 4095                                                                        & 4095                                                                          &                      & 0.24                     & 4095                                                                        &                      & 0.24                     & 4095                                                                        &                      & 0.36                     \\
qft\_12        & 12                                                   & 78                                                 &                      & 0.08                     & 4096                                                                        & 4096                                                                          &  & 1.27                     & 8191                                                                        & 8191                                                                          &                      & 0.45                     & 8191                                                                        &                      & 0.45                     & 8191                                                                        &                      & 1.55                     \\
qft\_13        & 13                                                   & 91                                                 &                      & 0.21                     & 8192                                                                        & 8192                                                                          &  & 2.88                     & 16383                                                                       & 16383                                                                         &                      & 0.82                     & 16383                                                                       &                      & 0.85                     & 16383                                                                       &                      & 2.85                     \\
qft\_14        & 14                                                   & 105                                                &                      & 0.39                     & 16384                                                                       & 16384                                                                         &  & 5.08                     & 32767                                                                       & 32767                                                                         &                      & 1.83                     & 32767                                                                       &                      & 1.81                     & 32767                                                                       &                      & 51.11                    \\
qft\_15        & 15                                                   & 120                                                &                      & 0.78                   & 32768                                                                       & 32768                                                                         &  & 9.06                     & 65535                                                                       & 65535                                                                         &                      & 3.69                     & 65535                                                                       &                      & 3.83                     & 65535                                                                       &                      & MO                       \\
qft\_16        & 16                                                   & 136                                                &                      & 1.8                    & 65536                                                                       & 65536                                                                         &  & 17.68                    & 131071                                                                      & 131071                                                                        &                      & 7.61                     & 131071                                                                      &                      & 7.37                     & 131071                                                                      &                      & MO                       \\
qft\_17        & 17                                                   & 153                                                &                      & 4.74                    & 131072                                                                      & 131072                                                                        &  & 36.50                    & 262143                                                                      & 262143                                                                        &                      & 15.09                    & 262143                                                                      &                      & 14.58                    & 262143                                                                      &                      & MO                       \\
qft\_18        & 18                                                   & 171                                                &                      & 13.65                   & 262144                                                                      & 262144                                                                        &  & 73.64                    & 524287                                                                      & 524287                                                                        &                      & 29.67                    & 524287                                                                      &                      & 30.55                    & 524287                                                                      &                      & MO                       \\
qft\_19        & 19                                                   & 190                                                &                      & 38.55                   & 524288                                                                      & 524288                                                                        &  & 141.63                   & 1048575                                                                     & 1048575                                                                       &                      & 60.82                    & 1048575                                                                     &                      & 63.54                    & 1048575                                                                     &                      & MO                       \\
qft\_20        & 20                                                   & 210                                                &                      & 137.98                  & 1048576                                                                     & 1048576                                                                       &  & 332.00                   & 2097151                                                                     & 2097151                                                                       &                      & 132.99                   & 2097151                                                                     &                      & 131.39                   & 2097151                                                                     &                      & MO                       \\
qft\_21        & 21                                                   & 231                                                &                      & 526.57                  & 2097152                                                                     & 2097152                                                                       &  & 2546.23                  & 4194303                                                                     & 4194303                                                                       &                      & 274.99                   & 4194303                                                                     &                      & 269.77                   & 4194303                                                                     &                      & MO                       \\ 
qft\_22 &22&253&&2102.72&4194304&4194304&&TO&-&-&&1668.13&	8388607&&1447.21 &	8388607&&MO\\


\hline\hline\hline
sum(-TO)            &                                                      &                                                    &                      & 809.24                  & 4240260                                                                     & 4230915                                                                       &  & 3871.86 
                  & 8515357                                                                     & 8493051                                                                       &                      & 1130.09                    &8496932
                                                                     &                      & 1147.50
                   & 8496488                                                                     &                      & -  
                   \\ \hline
\end{tabular}
}
}

\label{full-experiment-results}
\end{table}
}

\subsubsection{Compare with matrix-based methods}

As mentioned before, matrix-based methods, like Qiskit and TensorNetwork, represent an $n$-qubit circuit by a $2^{n} \times 2^{n}$ matrix. Assume that all data in such a matrix is represented in data type $complex128$. Then 64GB of memory must be allocated for the matrix of a 16-qubit circuit. This implies in particular that in our laptop (with 8GB RAM) these methods can process quantum circuits with at most 15 qubits. This observation is confirmed by our experiments. 
%
In comparison, the DD-based methods are often much more compact. \xblue{Indeed, for the qft circuits and on our laptop, both QMDD and TDD can process quantum circuits with up to 22 qubits. For example, both DDs can generate the functionality of the circuit `qft\_22' by using less than 256 MB memory.} For the bv circuits, this characteristic of DDs is even striking, as both DDs can easily process bv circuits with as many as 100 qubits in a few seconds! 

On the other hand, when the number of qubits is small ($\leq 10$), TensorNeteork usually works faster than DD-based methods. 
This is perhaps due to that the time-consumption for transforming matrices to decision diagrams is not negligible. 
\xblue{Surprisingly, while it takes only 0.08s for TensorNetwork to compute the tensor of the circuit `qv\_n9\_d5' (with 9 qubits and 200 gates), QMDD uses nearly 1000 seconds and TDD (with no partition) times out. For qv circuits, during the construction process, both QMDD and TDD may generate decision diagrams with $\mathcal{O}(2^{2n})$ nodes. 
Thus, it has a worst case time complexity $\mathcal{O}(2^{8n})$, as contracting two TDDs $\F$ and $\g$ needs $O(|\F|^2\cdot |\g|^2)$  time in the worst case (cf. Sec.~\ref{sec:contraction}). In comparison, for TensorNetwork, it uses optimised matrix multiplication and needs at most $\mathcal{O}(2^{6n})$ time.} 
It turns out that the TDD representation of `qv\_n9\_d5' has 262,144 nodes, while in comparison the TDD of `qft\_17' (with 17 qubits and 153 gates) has 262,143 nodes. This suggests that some quantum circuits can be better processed by TensorNetwork than DD-based methods and there are quantum circuits which may have no compact DD representations. 






\subsubsection{Compare with QMDD}
Recall that every non-terminal node in a TDD has two successors while any non-terminal node in a QMDD has four. In principle, the TDD representation of a quantum circuit has, \xblue{as we have seen in Section~\ref{sec:qmdd_vs._tdd}, 1-3 times the number of nodes as the circuit's QMDD, provided that the same order is used. Thus, the memory usage of the TDD representation is at most three times of that of the QMDD representation.} This observation is consistent with our experimental results. As a consequence, the TDD representation is as compact as QMDD.


\begin{figure}
	\centering
	\includegraphics[width=0.55\textwidth]{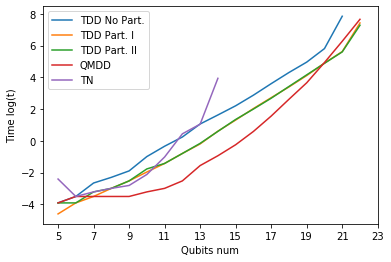}
	\caption{The logarithmic time consumption for constructing the functionality of qft circuits as the number of qubits increases, where  timeout is set as 3600s.}
	\label{qft_data}
\end{figure}

As far as runtime efficiency is concerned, \xblue{on the RevLib benchmarks from \cite{AstarZulehner}, the runtimes of the three TDD schemes are 38-110 times of that of QMDD; but, when considering the real quantum bv, qv, and qft circuits, the runtimes of the three TDD schemes are about 4.5, 1.3, 1.3 times of that of QMDD.} Considering that the TDD package is implemented in Python and QMDD is implemented in C++, this suggests that TDD, after further optimisation, has the potential to be comparable with QMDD. Actually, as can been observed from Fig.~\ref{qft_data}, for qft circuits with 20 or more qubits, the TDD package with either partition is already faster than QMDD.

\subsubsection{Compare among TDD schemes}

In general, the two partition schemes can significantly decrease the time-consumption for constructing the functionality of a quantum circuit.  From Table~\ref{experiment-results} we can see that both partition schemes can decrease the time-consumption by at least $50\%$ when compared with the no-partition scheme. This judgement is also confirmed by experiments on bv, qft, and qv circuits (cf. Table~\ref{full-experiment-results}).

Table~\ref{experiment-results} also suggests that the TDD construction with either partition scheme often has smaller intermediate diagrams than QMDD and the no-partition TDD. Let
\begin{align*}
\alpha = \frac{\mbox{maximum size of all DDs during the computing process}}{\mbox{the size of the final DD} }.
\end{align*}
Table~\ref{experiment-results} shows that the $\alpha$ values of QMDD and the three TDD schemes are, respectively, 2.46, 2.61, 1.28 and 1.25. That is, the ratio could be halved if either partition scheme is adopted.




\subsection{Summary and Discussion}
From the above empirical results we can see that 
\begin{itemize}
	\item TDD is compact and memory-saving and can be used for calculating the functionality of large circuits.
	\item TDD is flexible and can be easily combined with tensor network techniques (e.g., partition) to further improve its performance.
\end{itemize}


Besides representing the functionality of quantum circuits, TDD can also be used in the classical simulation of quantum circuits.
Experimental results show that the performance of TDD is similar to that of QMDD reported in \cite{zulehner2018advanced}. For example, we can obtain all amplitudes of  `qft\_$k$' circuits within 3 seconds for $k\leq 64$. We also conducted experiments on the simulation of random quantum circuits. The performance is also similar to that of QMDD.

In addition, TDD can also be used to calculate the trace of a quantum circuit, which plays a central role in calculating fidelity and hence checking if two quantum circuits are  approximately equivalent. \xblue{
In our recent work \cite{hong2021approximate}, two noisy quantum circuits, described by two super operators $\mathcal{E}_1, \mathcal{E}_2$ respectively, are said to be \emph{approximate equivalent} if $F_J(\mathcal{E}_1,\mathcal{E}_2)>1-\epsilon$, where $\epsilon \in [0,1]$ is a given value and $F_J$ is the Jamiolkowski fidelity \cite{Gil05}. Suppose one of the circuits is \emph{ideal} and its functionality is described by a unitary $U$, and the other is \emph{noisy} and its functionality is described by a super operator $\mathcal{E}=\{E_k\}_{k\in K}$. Then $F_J(\mathcal{E},U)=\frac{1}{2^{2n}}\sum_{k\in K}{|\textsf{Tr}(U^{\dagger}E_k)|^2}$. Each trace term can be calculated by contracting a tensor network obtained by connecting the input and output of each qubit in a miter constructed from the two quantum circuits. 
As demonstrated in \cite{hong2021approximate}, this task can be excellently completed by using TDD. The current QMDD package does not provide a customised procedure for trace calculation; instead, we need to get the QMDD representation first. Take the (ideal) `qv\_n9\_d5' circuit for example. Using our TDD package, it takes less than 2 seconds to calculate its trace \blue{(which does not require the calculation of the TDD representation of the circuit)}, while computing its QMDD requires nearly 1,000 seconds (cf. Table~\ref{full-experiment-results}). In comparison with matrix-based methods, experiments also show that the TDD method outperforms Qiskit and TensorNetwork, especially when the number of qubits is more than seven. We refer the reader to \cite{hong2021approximate} for more details.\footnote{\blue{It is worth noting that the results reported in \cite[Table~1]{hong2021approximate} was based on a previous implementation of TDD and the current implementation can be an order of magnitude faster.} } } 

As a direct extension of BDD from Boolean functions to tensors, TDDs can also represent classical gates.
More important, we can also represent the measurements and classically controlled gates as TDDs, which makes it suitable for coping with tasks such as equivalence checking of dynamic quantum circuits \cite{hong2021equivalence}. \xblue{The reader is also referred to \cite{Burgholzer21dyn} for the QMDD approach for the same task.}

\section{Conclusion} \label{sec:conclusion}

We proposed a decision diagram style data structure --- TDD --- for more principled applications of tensor networks. Based on a Boole-Shannon style expansion for tensors, it is rigorously proved that TDD provides a universal and canonical representation for quantum functionalities. As a decision diagram, TDD is also compact as redundant or isomorphic nodes have been completely removed or merged.
Experiments on a variety of benchmark circuits include qft confirm that TDD is compact, demonstrate its efficiency, and show that it often outperforms the Google TensorNetwork package for circuits with 15 or more qubits. Moreover, thanks to its origin from tensor networks, many techniques developed for tensor networks can be directly imported into TDD. As an example, we have shown that the TDD of a quantum circuit can be computed more efficiently by exploiting circuit partitions that were previously introduced for the classical simulation of quantum circuits.


It is expected that TDD can be used, possibly together with the Google TensorNetwork, in many design automation tasks, e.g., simulation and equivalence checking, for quantum circuits. In particular, we plan to combine TDD with TensorNetwork in our partition-based schemes. When the rank of the tensor is small, TensorNetwork runs faster than both decision diagrams; however, its performance decrease sharply when the number of qubits increases. Thus, we can use TensorNetwok to compute the (local) functionalities of each part, transform them into TDDs, and then contract these local TDDs to obtain the TDD representation of the quantum circuit.

The current TDD package is far from optimal. Future work will implement TDD in C++ and exploit more optimisation techniques developed in tensor networks, e.g., tree decomposition \cite{markov2008simulating}. 

In this paper, we assume that all indices of a tensor can only take values from $\{0,1\}$. This restriction can be removed by allowing a node in a TDD to have any number of successors. Moreover, different nodes can have different numbers of successors. What we should ensure is that all nodes corresponding to the same index have the same number of successors, and the contraction should be conducted on all its successors when this index is contracted. In our follow-up work, we plan to construct such a generalised package and make it suitable for more tensor network tasks.


\begin{acks}          
    \xblue{We thank the three reviewers for their valuable suggestions. The observation that each QMDD node corresponds to 1-3 TDD nodes (for some TDDs) was suggested by one of the reviewers. We also thank Lukas Burgholzer for suggestions on installing the QMDD package on a Windows platform.}
	This work is partially supported by the National Key R\&D Program of China (Grant No.: 2018YFA0306701) and the Australian Research Council (Grant No.s: DP180100691, DP220102059). 
\end{acks}


\bibliographystyle{acm}
\bibliography{references}

\begin{thebibliography}{10}

\bibitem{GoogleQsupr}
{\sc Arute, F., Arya, K., Babbush, R., Bacon, D., Bardin, J.~C., Barends, R.,
  Biswas, R., Boixo, S., Brandao, F.~G., Buell, D.~A., et~al.}
\newblock Quantum supremacy using a programmable superconducting processor.
\newblock {\em Nature 574}, 7779 (2019), 505--510.

\bibitem{bahar1997algebric}
{\sc Bahar, R.~I., Frohm, E.~A., Gaona, C.~M., Hachtel, G.~D., Macii, E.,
  Pardo, A., and Somenzi, F.}
\newblock Algebric decision diagrams and their applications.
\newblock {\em Formal Methods in System Design 10}, 2-3 (1997), 171--206.

\bibitem{bernstein1997quantum}
{\sc Bernstein, E., and Vazirani, U.}
\newblock Quantum complexity theory.
\newblock {\em SIAM Journal on Computing 26}, 5 (1997), 1411--1473.

\bibitem{biamonte2019lectures}
{\sc Biamonte, J.}
\newblock Lectures on quantum tensor networks.
\newblock {\em arXiv preprint arXiv:1912.10049\/} (2019).

\bibitem{boixo2017simulation}
{\sc Boixo, S., Isakov, S.~V., Smelyanskiy, V.~N., and Neven, H.}
\newblock Simulation of low-depth quantum circuits as complex undirected
  graphical models.
\newblock {\em arXiv preprint arXiv:1712.05384\/} (2017).

\bibitem{brace1990efficient}
{\sc Brace, K.~S., Rudell, R.~L., and Bryant, R.~E.}
\newblock Efficient implementation of a {BDD} package.
\newblock In {\em 27th ACM/IEEE Design Automation Conference\/} (1990), IEEE,
  pp.~40--45.

\bibitem{bryant1986graph}
{\sc Bryant, R.~E.}
\newblock Graph-based algorithms for boolean function manipulation.
\newblock {\em IEEE Transactions on Computers 100}, 8 (1986), 677--691.

\bibitem{bryant1995verification}
{\sc Bryant, Y.-A. C. R.~E.}
\newblock Verification of arithmetic circuits with binary moment diagrams.
\newblock In {\em 32nd Design Automation Conference\/} (1995), IEEE,
  pp.~535--541.

\bibitem{burgholzer2020advanced}
{\sc Burgholzer, L., and Wille, R.}
\newblock Advanced equivalence checking for quantum circuits.
\newblock {\em IEEE Transactions on Computer-Aided Design of Integrated
  Circuits and Systems\/} (2020).

\bibitem{burgholzer2020improved}
{\sc Burgholzer, L., and Wille, R.}
\newblock Improved dd-based equivalence checking of quantum circuits.
\newblock In {\em 2020 25th Asia and South Pacific Design Automation Conference
  (ASP-DAC)\/} (2020), IEEE, pp.~127--132.

\bibitem{Burgholzer21dyn}
{\sc Burgholzer, L., and Wille, R.}
\newblock Towards verification of dynamic quantum circuits.
\newblock {\em CoRR abs/2106.01099\/} (2021).

\bibitem{chen2018classical}
{\sc Chen, J., Zhang, F., Huang, C., Newman, M., and Shi, Y.}
\newblock Classical simulation of intermediate-size quantum circuits.
\newblock {\em arXiv preprint arXiv:1805.01450\/} (2018).

\bibitem{chen201864}
{\sc Chen, Z.-Y., Zhou, Q., Xue, C., Yang, X., Guo, G.-C., and Guo, G.-P.}
\newblock 64-qubit quantum circuit simulation.
\newblock {\em Science Bulletin 63}, 15 (2018), 964--971.

\bibitem{Gil05}
{\sc Gilchrist, A., Langford, N.~K., and Nielsen, M.~A.}
\newblock Distance measures to compare real and ideal quantum processes.
\newblock {\em Physical Review A 71}, 6 (2005), 062310.

\bibitem{gray2020hyper}
{\sc Gray, J., and Kourtis, S.}
\newblock Hyper-optimized tensor network contraction.
\newblock {\em arXiv preprint arXiv:2002.01935\/} (2020).

\bibitem{hong2021equivalence}
{\sc Hong, X., Feng, Y., Li, S., and Ying, M.}
\newblock Equivalence checking of dynamic quantum circuits.
\newblock {\em arXiv preprint arXiv:2106.01658\/} (2021).

\bibitem{hong2021approximate}
{\sc Hong, X., Ying, M., Feng, Y., Zhou, X., and Li, S.}
\newblock Approximate equivalence checking of noisy quantum circuits.
\newblock In {\em 2021 58th ACM/IEEE DAC\/} (2021), IEEE, pp.~1--6.

\bibitem{huang2020classical}
{\sc Huang, C., Zhang, F., Newman, M., Cai, J., Gao, X., Tian, Z., Wu, J., Xu,
  H., Yu, H., Yuan, B., et~al.}
\newblock Classical simulation of quantum supremacy circuits.
\newblock {\em arXiv preprint arXiv:2005.06787\/} (2020).

\bibitem{li2019quantum}
{\sc Li, R., Wu, B., Ying, M., Sun, X., and Yang, G.}
\newblock Quantum supremacy circuit simulation on sunway taihulight.
\newblock {\em IEEE Transactions on Parallel and Distributed Systems 31}, 4
  (2019), 805--816.

\bibitem{LiZF21}
{\sc Li, S., Zhou, X., and Feng, Y.}
\newblock Qubit mapping based on subgraph isomorphism and filtered
  depth-limited search.
\newblock {\em {IEEE} Trans. Computers 70}, 11 (2021), 1777--1788.

\bibitem{markov2008simulating}
{\sc Markov, I.~L., and Shi, Y.}
\newblock Simulating quantum computation by contracting tensor networks.
\newblock {\em SIAM Journal on Computing 38}, 3 (2008), 963--981.

\bibitem{MaslovFM08}
{\sc Maslov, D., Falconer, S.~M., and Mosca, M.}
\newblock Quantum circuit placement.
\newblock {\em {IEEE} Trans. on {CAD} of Integrated Circuits and Systems 27}, 4
  (2008), 752--763.

\bibitem{miller2006qmdd}
{\sc Miller, D.~M., and Thornton, M.~A.}
\newblock Qmdd: A decision diagram structure for reversible and quantum
  circuits.
\newblock In {\em 36th International Symposium on Multiple-Valued Logic
  (ISMVL'06)\/} (2006), IEEE, pp.~30--30.

\bibitem{molitor2007equivalence}
{\sc Molitor, P., and Mohnke, J.}
\newblock {\em Equivalence Checking of Digital Circuits: Fundamentals,
  Principles, Methods}.
\newblock Springer Science \& Business Media, 2007.

\bibitem{moll2018quantum}
{\sc Moll, N., et~al.}
\newblock Quantum optimization using variational algorithms on near-term
  quantum devices.
\newblock {\em Quantum Science and Technology 3}, 3 (2018), 030503.

\bibitem{nielsen2002quantum}
{\sc Nielsen, M.~A., and Chuang, I.~L.}
\newblock {\em Quantum Computation and Quantum Information (10th Anniversary
  edition)}.
\newblock Cambridge University Press, 2016.

\bibitem{niemann2015qmdds}
{\sc Niemann, P., Wille, R., Miller, D.~M., Thornton, M.~A., and Drechsler, R.}
\newblock {QMDD}s: Efficient quantum function representation and manipulation.
\newblock {\em IEEE Transactions on Computer-Aided Design of Integrated
  Circuits and Systems 35}, 1 (2015), 86--99.

\bibitem{pednault2017breaking}
{\sc Pednault, E., Gunnels, J.~A., Nannicini, G., Horesh, L., Magerlein, T.,
  Solomonik, E., Draeger, E.~W., Holland, E.~T., and Wisnieff, R.}
\newblock Breaking the 49-qubit barrier in the simulation of quantum circuits.
\newblock {\em arXiv preprint arXiv:1710.05867\/} (2017).

\bibitem{roberts2019tensornetwork}
{\sc Roberts, C., Milsted, A., Ganahl, M., Zalcman, A., Fontaine, B., Zou, Y.,
  Hidary, J., Vidal, G., and Leichenauer, S.}
\newblock Tensornetwork: A library for physics and machine learning.
\newblock {\em arXiv preprint arXiv:1905.01330\/} (2019).

\bibitem{tsai2020bit}
{\sc Tsai, Y.-H., Jiang, J.-H.~R., and Jhang, C.-S.}
\newblock Bit-slicing the {H}ilbert space: Scaling up accurate quantum circuit
  simulation to a new level.
\newblock {\em arXiv preprint arXiv:2007.09304\/} (2020).

\bibitem{viamontes2003improving}
{\sc Viamontes, G.~F., Markov, I.~L., and Hayes, J.~P.}
\newblock Improving gate-level simulation of quantum circuits.
\newblock {\em Quantum Information Processing 2}, 5 (2003), 347--380.

\bibitem{vinkhuijzen2021limdd}
{\sc Vinkhuijzen, L., Coopmans, T., Elkouss, D., Dunjko, V., and Laarman, A.}
\newblock {LIMDD}: A decision diagram for simulation of quantum computing
  including stabilizer states.
\newblock {\em arXiv preprint arXiv:2108.00931\/} (2021).

\bibitem{wille2020efficient}
{\sc Wille, R., Hillmich, S., and Burgholzer, L.}
\newblock Efficient and correct compilation of quantum circuits.
\newblock In {\em IEEE International Symposium on Circuits and Systems\/}
  (2020).

\bibitem{yamashita2008ddmf}
{\sc Yamashita, S., Minato, S.-i., and Miller, D.~M.}
\newblock {DDMF}: An efficient decision diagram structure for design
  verification of quantum circuits under a practical restriction.
\newblock {\em IEICE Transactions on Fundamentals of Electronics,
  Communications and Computer Sciences 91}, 12 (2008), 3793--3802.

\bibitem{AstarZulehner}
{\sc Zulehner, A., Paler, A., and Wille, R.}
\newblock An efficient methodology for mapping quantum circuits to the {IBM}
  {QX} architectures.
\newblock {\em IEEE Transactions on Computer-Aided Design of Integrated
  Circuits and Systems 38}, 7 (2018), 1226--1236.

\bibitem{zulehner2018advanced}
{\sc Zulehner, A., and Wille, R.}
\newblock Advanced simulation of quantum computations.
\newblock {\em IEEE Transactions on Computer-Aided Design of Integrated
  Circuits and Systems 38}, 5 (2018).

\end{thebibliography}

\subsection*{Appendix: Detailed Proofs}\label{sec:appendix}
\begin{lemma*}[Lemma~\ref{lem:tensor-phase}]
For any tensor $\phi$ which is not normal, there exists a unique normal tensor $\phi^*$ such that $\phi = p\cdot \phi^*$, where $p$ is a nonzero complex number.
\end{lemma*}
\begin{proof}
Since $\phi$ is not normal, we have $\phi\not=0$. Let $p=\phi(\vec{a}^*)$ where $\vec{a}^*$ is the pivot of $\phi$. Then 
$\phi^* := \frac{1}{p} \cdot \phi$ is a normal tensor which satisfies the condition. Furthermore, suppose there is another normal tensor $\phi'$ such that $\phi = p'\cdot \phi'$ for some complex number $p'$. Then we have $\phi = p \cdot \phi^* = p' \cdot \phi'$. Obviously, 
\rmagenta{we have $|p|=|p'|$ and, by definition,}
$\phi^*$ and $\phi'$ also share the same pivot $\vec{a}^*$ with $\phi$. It then follows that $\phi^*(\vec{a}^*) = \phi'(\vec{a}^*) =1$. Thus $p = p'$, and $\phi^* = \phi'$.
\end{proof}

\begin{lemma*}\label{lem:normalTDDedge}
Every terminal node of a normal TDD $\F$ has value 0 or 1. Moreover, let $v$ be a non-terminal node of $\F$ with $\bphi(v)\neq 0$, and $w_0$ and $w_1$ the weights on its low- and high-edge. Then  we have either $w_0=1$ or $w_1=1$.
\end{lemma*}
\begin{proof} 
The terminal case is clear by definition.
Suppose the index set of $\F$ is $\{x_1, \ldots, x_n\}$ and  $x_1\prec \ldots \prec x_n$.
 For a non-terminal node $v$,
let $\phi$, $\phi_l$, and $\phi_h$ denote $\bphi(v), \bphi(low(v))$, and $\bphi(high(v))$, respectively. 
Then $\phi=w_0 \cdot \overline{x}_i \cdot \phi_l + w_1 \cdot x_i \cdot \phi_h$ by Eq.~\ref{eq:tdd-intp}, where $x_i = index(v)$.  
Note that $\phi$ is a tensor over $\{x_i, \ldots, x_n\}$ and both $\phi_l$ and $\phi_h$ can be regarded as tensors over $\{x_{i+1}, \ldots, x_n\}$. 

Let $\vec{a}^*$ be the pivot of $\phi$. 
 Suppose $\vec{a}^* = 0\vec{b}^*$ for some $\vec{b}^*\in \{0,1\}^{\rmagenta{n-i}}$; that is,  $\vec{a}^*$ takes value 0 at index $x_i$. Then by $1= \phi(\vec{a}^*)= w_0 \cdot \phi_l(\vec{b}^*)$, we have 
 $|w_0| \geq 1$ from the fact that $\phi_l$ is normal. On the other hand,
let $\vec{c}$ be the pivot of $\phi_l$. Then from
 $\phi(0\vec{c})= w_0 \cdot \phi_l(\vec{c}) = w_0$ and the fact that $\phi$ is normal, we have $|w_0| \leq 1$. Thus $|w_0| = 1$ and $|\phi_l(\vec{b}^*)|=1$.
  Now for any $\vec{b}\in \{0,1\}^{\rmagenta{n-i}}$ which is less than $\vec{b}^*$ in the lexicographic order, we have $|\phi_l(\vec{b})| =  |\phi(0\vec{b})|< |\phi(\vec{a}^*)|=1$, as $0\vec{b}$ is less than $0\vec{b}^*=\vec{a}^*$. Thus by definition, $\vec{b}^*$ is actually the pivot of $\phi_l$. So $\phi_l(\vec{b}^*)=1$ and hence $w_0=1$. 
  
  The case when $\vec{a}^*$ takes value 1 at index $x_i$ is analogous.
\end{proof}

\begin{theorem*}[Theorem~\ref{thm:complete}]
Let $I = \{x_1,x_2,...,x_n\}$ be a set of indices and $\prec$ a linear order on it. For any tensor $\phi$ with index set $I$, there exists a $\prec$-ordered normal TDD $\F$ such that $\bphi(\F) = \phi$.
\end{theorem*}
\begin{proof}
We prove the result by induction on the cardinality of $I$. If $|I|=0$, the tensor is a constant and the conclusion clearly holds after possible application of NR1. Suppose the statement holds for tensors with up to $n$ indices. We show it is also true for tensors with $n+1$ indices. Let $I=\{x_1, ..., x_{n+1}\}$ be the index set and, without loss of generalisation, assume $x_1\prec x_2 \prec ...\prec x_{n+1}$. Given an arbitrary tensor $\phi$ over $I$, by the Boole-Shannon expansion, we know 
\begin{align*}
\phi = \overline{x}_1 \cdot \phi_{0} + x_1 \cdot \phi_{1},
\end{align*}
where $\phi_c := \phi|_{x_1=c}$ for $c\in \{0,1\}$.
Since $\phi_{c}$ is a tensor over $n$ indices, by induction hypothesis, there is a $\prec'$-ordered normal TDD $\F_c$ such that $\phi_{c} =  \bphi(\F_c)$, where $\prec'$ is the restriction of $\prec$ on $I\setminus\{x_1\}$.  Let $r_{c}$ be the root node and  $w_c := w_{\F_c}$ the weight of $\F_c$. Then we have $\phi_{c} =  \bphi(\F_c)= w_c \cdot \bphi(r_{c})$. Next, we introduce a new root node $v$ with weight 1 on its incoming edge. Set $low(v)$ and $high(v)$ to be $r_0$ and $r_1$ respectively. Furthermore, set the weights on the low- and high-edges of $v$ to be $w_{0}$ and $w_{1}$, respectively. The constructed TDD, denoted by $\F$, is $\prec$-ordered and, after applying the normalisation rule NR2 on $v$, normal. By Eq.~\ref{eq:tdd-intp}, we have $\bphi(\F)=\phi$. 
\end{proof}

\begin{lemma*}[Lemma~\ref{lem: nodes_in_reducedTDD_are_essential}]
Suppose $\F$ is a reduced TDD of a non-constant tensor $\phi$ over index set $I$. Then every non-terminal node of $\F$ is labelled with an index that is essential to $\phi$.  
\end{lemma*}
\begin{proof}
	Suppose $v$ is a non-terminal node of $\F$ which is labelled with a non-essential index $x$. Let $\phi'=\bphi(v)$. Then $\phi'|_{x=0}= \phi'|_{x=1}$. From Eq.~\ref{eq:tddnode},
	$\phi'|_{x=0}= w_0\cdot \bphi(low(v))$ and $\phi'|_{x=1}= w_1\cdot \bphi(high(v))$, where $w_0$ and $w_1$ are the weights on the low- and high-edges of $v$, respectively. It follows by Lemma~\ref{lem:tensor-phase} that $\bphi(low(v))=\bphi(high(v))$ and $w_0=w_1$
	since they are both normal. Note that $low(v))$ and $high(v)$ may be identical.  From  Lemma~\ref{lem:normalTDDedge}, we have $w_0=w_1=1$ and thus $\bphi(v) = \overline{x} \cdot \bphi(low(v)) + x \cdot \bphi(high(v))=\bphi(low(v))$. This shows that we have two nodes, viz. $v$ and $low(v)$, representing the same tensor, which contradicts the assumption that $\F$ is reduced.
\end{proof}

\begin{theorem*}[canonicity, Theorem~\ref{thm:canonicity}]
Let $I$ be an index set and $\prec$ a linear order on $I$. Suppose $\F$ and $\g$ are two  $\prec$-ordered, reduced TDDs over $I$ with $\bphi(\F)=\bphi(\g)$. Then $\F\eqsim\g$. 
\end{theorem*}
\begin{proof}
We prove this by induction on the cardinality of $I$. First,  reduced TDDs of any constant tensor are clearly unique. In particular, from 1) and 2) of Definition~\ref{def:recuded}, the 0 tensor is represented by the reduced TDD with weight 0 which has a unique node, viz. terminal 1.

Suppose the statement holds for any tensor with at most $n$ indices. We prove it also holds for tensors with $n+1$ indices. From $\bphi(\F)=\bphi(\g)$, we have by Lemma~\ref{lem:sameNormalTDDs} that $\bphi(r_\F)=\bphi(r_\g)$ and $w_\F=w_\g$.
In addition, by Lemma~\ref{lem: nodes_in_reducedTDD_are_essential}, $r_\F$ and $r_\g$ are labeled with essential indices. They must be the same as, otherwise, the smaller one in the order $\prec$ is not essential for either $\F$ or $\g$. Let $x$ be this variable. By Lemma~\ref{lem:tdd-shannon}, we have
\begin{align*}
\bphi(\F) &=\overline{x} \cdot \bphi(\F|_{x=0}) + x \cdot \bphi(\F|_{x=1})\\
\bphi(\g) &=\overline{x} \cdot \bphi(\g|_{x=0}) + x \cdot \bphi(\g|_{x=1}).
\end{align*}
Since $\bphi(\F)=\bphi(\g)$, it holds that $\bphi(\F|_{x=c})=\bphi(\g|_{x=c})$ for $c\in \{0,1\}$. By the induction hypothesis, we have $\F_c\eqsim\g_c$. This, together with the fact that $index(r_\F)=index(r_\g)$, implies that $\F\eqsim\g$.
\end{proof}

\begin{theorem*}[Theorem~\ref{thm:reduced_if_no_rule_applicable}]
A normal TDD is reduced if and only if no reduction rule is applicable. 
\end{theorem*}
\begin{proof}
Clearly, a normal TDD $\F$ is not reduced if at least one reduction rule is applicable as, otherwise, we shall have either a node representing tensor 0 or two nodes representing the same tensor. 

On the other hand, suppose no reduction rule is applicable to $\F$. We show by induction on the depth of $\F$ that  $\F$ is reduced. First, from the fact that RR1 and RR2 are not applicable, $\F$ must have a unique terminal node with value 1, and all edges weighted 0 have been redirected to it.

Assume there exist  non-terminal nodes which represent tensor 0 and $v$ is such a node with the $\prec$-largest label. By our assumption and that $label(v)\prec label(low(v))$ and $label(v)\prec label(high(v))$, we have $\bphi(low(v))\not=0$ and $\bphi(high(v))\not =0$. Now, as $\bphi(v)=0$, the weights on the low- and high-edges of $v$ must both be 0, which however activates either RR2 or RR3 and thus a contradiction with our assumption. 

Suppose there are two non-terminal nodes $v$ and $v'$ with $\bphi(v)=\bphi(v')$. Let $\F_v$ and $\F_{v'}$ be the sub-TDDs of $\F$ rooted at $v$ and $v'$ respectively (but set their weights to be 1). Note that no reduction rule is applicable to either $\F_v$ or $F_{v'}$, since otherwise it is also applicable to $\F$.
Then by induction hypothesis, they are both reduced. Furthermore, we have $\bphi(\F_v)=\rmagenta{\bphi(v) =\bphi(v')=}\bphi(\F_{v'})$, and from Theorem~\ref{thm:canonicity}, $\F_v \eqsim \F_{v'}$. Then $index(v) = index(v')$ and $w_0=w_0'$, where $w_0$ and $w_0'$ are the weights on the low-edges of $v$ and $v'$, respectively. Furthermore, it follows from Eq.~\ref{eq:tdd-intp} that $\bphi(low(v)) = \bphi(low(v'))$. By induction hypothesis, we have $low(v) = low(v')$. Similarly, we can prove that $high(v) = high(v')$ as well. That is, RR4 is applicable to $v$ and $v'$ \rmagenta{and thus also a contradiction with our assumption.}

In summary, we have shown that  $\F$ is reduced.
\end{proof}

\end{document}